%% file: head.tex
\documentclass[sigconf]{acmart}
\acmDOI{x}
\setcopyright{acmlicensed}\acmConference[]{}{}{}
\copyrightyear{2022}
\acmYear{2022}

\usepackage{booktabs} 
\usepackage{graphicx}
\usepackage{ifthen}

\usepackage[utf8]{inputenc} 
\usepackage[T1]{fontenc}    
\usepackage{url}            
\usepackage{booktabs}       
\usepackage{amsfonts}       
\usepackage{dsfont}

\usepackage{nicefrac}       
\usepackage{microtype}      
\usepackage{amsmath}
\usepackage{amsthm}
\usepackage{bm}
\usepackage{subcaption}
\usepackage{mathrsfs}
\usepackage{xcolor}
\usepackage{xspace}
\usepackage{enumitem}
\usepackage{algorithm}
\usepackage{algcompatible}
\algnewcommand\algorithmicreturn{\textbf{return }}
\algnewcommand\RETURN{\State \algorithmicreturn}%

\usepackage{multirow}
\usepackage{hyperref}
\usepackage[colorinlistoftodos,prependcaption]{todonotes}
\usepackage{natbib}
\usepackage{balance}
\usepackage{mathtools}
\input{macros}
\input{preamble}

\begin{document}
\title[Fair Matrix Factorisation for Large-Scale Recommender Systems]{Fair Matrix Factorisation for Large-Scale Recommender Systems}


\author{Riku Togashi}
\affiliation{
  \institution{CyberAgent, Inc.}
  \city{Tokyo} 
  \country{Japan}
}
\email{rtogashi@acm.org}

\author{Kenshi Abe}
\affiliation{
  \institution{CyberAgent, Inc.}
  \city{Tokyo} 
  \country{Japan}
}
\email{abe_kenshi@cyberagent.co.jp}

\renewcommand{\shortauthors}{
}
\makeatletter
\newcommand{\quickwordcount}[1]{%
  \immediate\write18{texcount -1 -sum -merge #1.tex > #1-words}
  \immediate\openin\somefile=#1-words%
  \read\somefile to \@@localdummy%
  \immediate\closein\somefile%
  \setcounter{wordcounter}{\@@localdummy}%
  \@@localdummy%
}
\makeatother

\begin{abstract}
Recommender systems are hedged with various requirements, such as ranking quality, optimisation efficiency, and item fairness.
Item fairness is an emerging yet impending issue in practical systems.
The notion of item fairness requires controlling the opportunity of items (e.g. the exposure) by considering the entire set of rankings recommended for users.
However, the intrinsic nature of fairness destroys the separability of optimisation subproblems for users and items,
which is an essential property of conventional scalable algorithms, such as implicit alternating least squares (iALS).
Few fairness-aware methods are thus available for large-scale item recommendation.
Because of the paucity of simple tools for practitioners, unfairness issues would be costly to solve or, at worst, would be abandoned.
This study takes a step towards solving real-world unfairness issues by developing a simple and scalable collaborative filtering method for fairness-aware item recommendation.
We built a method named \emph{fiADMM}, which inherits the scalability of iALS and maintains a provable convergence guarantee.
\end{abstract}

\keywords{
recommender systems;
collaborative filtering;
fairness;
efficiency;
}

\maketitle
\input{body}

\appendix
\input{appendix_double_column}
\balance
\bibliographystyle{ACM-Reference-Format}
\bibliography{head.bib}

\end{document}

%% file: macros.tex
\newcommand{\norm}[1]{\left\lVert #1 \right\rVert}
\DeclarePairedDelimiterX{\ind}[1]{\mathbb{I}\{}{\}}{#1}

\DeclareMathOperator*{\argmin}{argmin}

\newcommand{\calU}{\mathcal{U}}
\newcommand{\calV}{\mathcal{V}}

\newcommand{\calO}{\mathcal{O}}

\newcommand{\R}{\mathbb{R}}
\newcommand{\w}{\mathbf{w}}

\newcommand{\bo}{\mathbf{o}}
\newcommand{\br}{\mathbf{r}}

\newcommand{\bfG}{\mathbf{G}}
\newcommand{\s}{\mathbf{s}}
\newcommand{\bt}{\mathbf{t}}

\newcommand{\bR}{\mathbf{R}}
\newcommand{\I}{\mathbf{I}}

\newcommand{\bfH}{\mathbf{H}}
\newcommand{\bfU}{\mathbf{U}}
\newcommand{\bfV}{\mathbf{V}}

\newcommand{\1}{\mathds{1}}
\newcommand{\bu}{\mathbf{u}}
\newcommand{\bv}{\mathbf{v}}

\newcommand{\bLambda}{\mathbf{\Lambda}}

\newcommand{\ML}{\textit{ML-20M}\xspace}
\newcommand{\MSD}{\textit{MSD}\xspace}

\usepackage[nameinlink]{cleveref}
\crefname{lemma}{Lemma}{Lemmas}
\crefname{theorem}{Theorem}{Theorems}
\crefname{claim}{Claim}{Claims}
\crefname{remark}{Remark}{Remarks}
\crefname{observation}{Observation}{Observations}
\crefname{corollary}{Corollary}{Corollaries}
\crefname{appendix}{Appendix}{Appendices}
\crefname{section}{Section}{Sections}
\crefname{algorithm}{Algorithm}{Algorithms}
\crefname{equation}{Eq.}{Eqs.}
\crefname{figure}{Figure}{Figures}
\crefname{table}{Table}{Tables}

%% file: preamble.tex
\usepackage{amsmath}

\allowdisplaybreaks[1]

%% file: body.tex
\section{Introduction}
Modern recommender systems have rather complex responsibilities, such as accountability, transparency, and fairness.
Considering that users often have dual roles as stakeholders (e.g. consumers and producers)~\cite{burke2017multisided,abdollahpouri2020multistakeholder},
fairness for items has become part of the overall user utility, and is also a social responsibility.
However, optimising recommender systems while considering item fairness is a challenge.
Apart from this responsibility, systems must comply with internal requirements, i.e. ranking quality and computational efficiency.
Implementing a practical system thus entails finely balancing these requirements and responsibilities.
Computational efficiency is particularly critical because it is a prerequisite for implementation.

To date, considerable research effort has been devoted to the development of scalable item recommendation~\cite{hu2008collaborative,zhou2008large,yu2014distributed,he2016fast,bayer2017generic}.
Scalability is a major challenge in maximising user utility in the sense of ranking quality, which is the primary responsibility of recommender systems.
\emph{Implicit alternating least squares (iALS)}~\cite{hu2008collaborative}\footnote{Following \citet{rendle2021revisiting}, we use the term iALS to refer to the method (including the objective and optimisation algorithm) proposed by \citet{hu2008collaborative}.} is the most efficient collaborative filtering method based on matrix factorisation (MF).
Even after a decade since its emergence, iALS is still competitive in terms of ranking quality with its unrivalled scalability~\cite{rendle2020neural,rendle2021revisiting}.
The key to its scalability is \emph{optimisation separability}, realised by its pointwise loss function and alternating optimisation strategy.
That is, when item latent factors are fixed, the optimisation problem for each user factor is an independent (i.e. embarrassingly parallelisable) linear regression with a closed-form solution.
The optimisation efficiency of iALS is fascinating, making it irreplaceable for large-scale applications.

From the viewpoint of optimisation separability, item fairness constraints in top-$K$ ranking problems are \emph{intrinsically} problematic.
Item fairness requires restrictions on all users and items,
because it involves the uniformity of exposure allocation to items under the limited budget of users' top-$K$ results.
This is a major distinction between fairness-aware item recommendation and fairness-agnostic settings---\emph{optimal item rankings for users depend on each other.}
However, this intrinsic dependency inevitably destroys the optimisation separability, and thus the computational cost of most conventional methods is prohibitively large, a point we review in this paper.

This work aims to develop a fairness-aware method of which the optimisation efficiency is comparable to that of iALS.
Our aim is to provide an easy-to-use and scalable tool to solve immediate unfairness issues in real-world applications.
To this end, we devise a variant of iALS to learn MF models by considering item fairness while maintaining scalability.
We first propose a tractable fairness regulariser, which remains challenging to optimise due to the intrinsic optimisation non-separability.
We then develop an algorithm to optimise the fairness-aware MF without sacrificing scalability, even under the proposed regulariser.
Furthermore, despite the non-convex and multi-block optimisation in the proposed objective,
we provide a convergence guarantee for our proposed algorithm based on the alternating direction method of multipliers (ADMM)~\cite{boyd2011distributed}.

\section{Preliminary}
\subsection{Implicit Alternating Least Squares (iALS)}
We first review iALS~\cite{hu2008collaborative} for discussing the inefficiency issues in conventional fairness-aware ranking methods.

Given users $\mathcal{U}=[|\calU|]$ and items $\mathcal{V}=[|\calV|]$,
let $\bR \in \{0, 1\}^{|\calU| \times |\calV|}$ be an implicit feedback matrix whose $(i,j)$-element has the value of 1 when user $i \in \calU$ has interacted with item $j \in \calV$ and otherwise 0; we represent the number of observed interactions by that of non-zero entries in $\bR$, that is, $\mathrm{nz}(\bR)$.
The model parameters of iALS are the $d$-dimensional latent factors $\bfU \in \R^{|\calU|\times d}$ and $\bfV \in \R^{|\calV|\times d}$ for the users and items, respectively.
These parameters are estimated by minimising the loss function of iALS, which is defined as follows:
\begin{align}
    L(\bfV, \bfU) &= \frac{1}{2}\norm{\bR \odot (\bR - \bfU\bfV^\top)}_F^2 + \frac{\alpha_0}{2}\norm{\bfU\bfV^\top}_F^2 \nonumber\\
    &+ \frac{1}{2}\norm{\bLambda_U^{1/2} \bfU}_F^2 + \frac{1}{2}\norm{\bLambda_V^{1/2}\bfV}_F^2,
    \label{eq:obj-imf}
\end{align}
where Operator $\odot$ is the Hadamard element-wise product, and
the second term is the implicit regulariser~\cite{bayer2017generic}, which is the L2 norm of the recovered score matrix $\bfU\bfV^\top$.
For the implicit regulariser, we use a weight parameter $\alpha_0 > 0$.
In L2 regularisation, $\bLambda_U \in \R^{|\calU| \times |\calU|}$ and $\bLambda_V \in \R^{|\calV| \times |\calV|}$ are diagonal matrices representing the per-coordinate weights for user and item factors.
It is well known that ranking performance can often be improved by using weights that depend on the number of interactions for each user and item~\cite{hu2008collaborative,rendle2021revisiting}.
Let $\br_{i}$ and $\br_{:,j}$ be the (column) vectors that correspond to the $i$-th row and $j$-th column of $\bR$, respectively.
The frequency-based strategy sets the weights with base weight $\lambda_2>0$ and exponent $\eta \geq 0$ as follows:
\begin{align*}
  (\bLambda_U)_{i,i} = \lambda_2 \left(\norm{\br_{i}}_1 + \alpha_0|\calV|\right)^{\eta}, \,\,\,(\bLambda_V)_{j,j} = \lambda_2 \left(\norm{\br_{:,j}}_1 + \alpha_0|\calU|\right)^{\eta}.
\end{align*}
Hereafter, we denote $\lambda_U^{(i)}=(\bLambda_U)_{i,i}$ and $\lambda_V^{(j)}=(\bLambda_V)_{j,j}$.

iALS solves the minimisation problem $\min_{\bfV,\bfU} L(\bfV, \bfU)$ by alternating the optimisation with respect to $\bfV$ and $\bfU$.
Specifically, in the $k$-th step, iALS updates $\bfV$ and $\bfU$ as follows:
\begin{align*}
  \bfU^{k+1} &= \argmin_{\bfU} \|\bR \odot (\bR - \bfU(\bfV^{k})^\top)\|_F^2 + \alpha_0\|\bfU(\bfV^{k})^\top\|_F^2 + \|\bLambda_U^{1/2}\bfU\|_F^2, \\
  \bfV^{k+1} &= \argmin_{\bfV}\|\bR \odot (\bR - {\bfU^{k+1}}\bfV^\top)\|_F^2 + \alpha_0\|{\bfU^{k+1}}\bfV^\top\|_F^2 + \|\bLambda_V^{1/2}\bfV\|_F^2.
\end{align*}
Owing to the alternating strategy, the optimisation for $\bfU$ and $\bfV$ can be divided into independent convex problems with respect to each row of $\bfU$ and $\bfV$.
Suppose that $\bu_i \in \R^d$ is the (column) vector that corresponds to the $i$-th row of $\bfU$.
Then, the update of $\bu_i$ is the following row-wise independent problem:
\begin{align*}
  \bu_i &= \argmin_{\bu_i} \norm{\br_i \odot (\br_i - \bfV\bu_i)}_2^2 + \alpha_0\norm{\bfV\bu_i}_2^2 + \lambda_U^{(i)}\norm{\bu_i}_2^2 \\
  &= \left(\sum_{j \in \calV}r_{i,j} \bv_j \bv_j^\top + \alpha_0\bfG_V + \lambda_U^{(i)}\I\right)^{-1}\sum_{j \in \calV}r_{i,j}\bv_j,
\end{align*}
where $\bfG_V = \sum_{j \in \calV}\bv_j {\bv_j}^\top = \bfV^\top\bfV$ is the Gramian matrix of item latent factors, where $\bv_j \in \R^{d}$ denotes the column vector that corresponds to the $j$-th row of $\bfV$.
When $\bfG_V$ is pre-computed, the expected computational cost for each subproblem is reduced to $\calO((\mathrm{nz}(\bR)/|\calV|)d^2 + d^3)$ (a.k.a. the Gramian trick~\cite{rendle2021revisiting}), which is realised by (1) computing the Gramian for the interacted items $\sum_{j \in \calV}r_{i,j}\bv_j\bv_j^\top$ in $\calO((\mathrm{nz}(\bR)/|\calV|)d^2)$ and by (2) solving the linear system $\bfH\bu_i\!=\!\sum_{i \in \calU}r_{i,j}\bv_j$, where $\bfH=\sum_{j \in \calV}r_{i,j}\bv_j\bv_j^\top\!+\!\alpha_0 \bfG_V \!+\!\lambda_U^{(i)}\I$ in $\calO(d^3)$.
Because the update of $\bfV$ is analogous to that of $\bfU$,
the overall cost of updating $\bfU$ and $\bfV$ is $\calO(\mathrm{nz}(\bR)d^2 + (|\calU|+|\calV|)d^3)$.
This is much faster than $\calO(|\calU||\calV|d^2 + (|\calU|+|\calV|)d^3)$ owing to feedback sparsity $\mathrm{nz}(\bR)\ll|\calU||\calV|$.

In summary, iALS retains scalability, despite its objective involves \emph{all} user-item pairs due to the implicit regulariser.
The crux is that iALS avoids the intractable factor $\calO(|\calU||\calV|)$ owing to the Gramian trick and feedback sparsity.

\subsection{Inefficiency Issue in Fair Ranking}
\label{section:inefficiency-in-fair-ranking}
Considering the above discussion, we here review the inefficiency issue with conventional fairness-aware ranking methods without any distinction among the types of fairness (e.g. user/item or group/individual).

Numerous studies have adopted an approach to learn fair probabilistic ranking policies based on \emph{given} preferences~\cite{biega2018equity,singh2018fairness,memarrast2021fairness,do2021two,wu2021multi,do2022optimizing,saito2022fair}.
The optimisation is often formulated as a convex optimisation on $|\calV| \times |\calV|$ (or $|\calV| \times K$ for top-$K$ ranking) doubly stochastic matrices for each user with fairness constraints.
This approach may not apply to realistic recommender systems owing to the $\calO(|\calU||\calV|^2)$ (or $\calO(|\calU||\calV|K)$) space complexity for the parameters to be optimised.
Scalability can be enhanced by reformulating the subproblem for each user as an ADMM~\cite{boyd2011distributed}, which iteratively solves the local optimisation problems for two $|\calV| \times |\calV|$ row-wise or column-wise stochastic matrices and a $|\calV| \times |\calV|$ dual variable\footnote{For details, see the appendix of \citet{memarrast2021fairness}.}.
To avoid the $\calO(|\calU||\calV|^2)$ cost,
recent methods~\cite{do2021two,usunier2022fast,do2022optimizing} rely on the Frank-Wolfe algorithm~\cite{Frank1956AnAF,jaggi2013revisiting}, which requires top-$K$ sorting of items for each user at each iteration, resulting in a computational cost of $O(|\calU||\calV| \log K)$ per epoch; this is still prohibitively large in real-world applications.
\citet{patro2020fairrec} proposed the greedy-round-robin algorithm, which also does not scale well because its round-robin scheduling is not parallelisable with respect to users.
This \emph{post-processing} approach requires, a priori, a $|\calU|\times|\calV|$ (dense) preference matrix (e.g. $\bfU\bfV^\top$ of an MF model), which is costly to retain in the memory space and even impossible to explicitly compute owing to the cost of $\calO(|\calU||\calV|C)$, where $C$ is the cost for predicting a single user-item pair (e.g. $C=d$ for MF).
Therefore, the post-processing approach cannot exploit feedback sparsity, leading to the computational cost of $O(|\calU||\calV|)$.
It should also be noted that the fairness guarantees proved in the previous studies hold when the true examination probabilities for all rank positions and true preferences for all user-item pairs are known; hence, the guarantees do not hold in practise because only estimates with errors are usually accessible.

In contrast to the post-processing approach, various studies have explored its \emph{in-processing} counterpart in which a single model is trained to optimise its ranking quality and fairness simultaneously~\cite{kamishima2011fairness,kamishima2013efficiency,yao2017beyond,burke2018balanced,singh2019policy,zehlike2020reducing,morik2020controlling,yadav2021policy,oosterhuis2021computationally}.
Most methods are designed for re-ranking tasks in information retrieval, where the number of items (documents) is typically small.
To represent a stochastic ranking policy,
several studies in the context of information retrieval~\cite{singh2019policy,yadav2021policy,oosterhuis2021computationally} use the Placket-Luce model~\cite{plackett1975analysis}, of which the cost is $\calO(|\calU||\calV|K)$ per epoch.
Moreover, mini-batch stochastic gradient descent (SGD) is adopted for optimising the multiple objectives of ranking and fairness.
Although mini-batch SGD allows us to use flexible models/objectives and reduces the computational cost in a single step, it is disadvantaged by slow convergence.
By contrast, \citet{burke2018balanced} proposed fairness-aware variants of SLIM~\cite{ning2011slim}.
Because their group-based fairness regularisers retain the optimisation separability with respect to the rows (or columns) of the weight matrix as in the original SLIM,
their approach enables parallel optimisation based on coordinate descent.
However, extending their group-wise approach to individual item fairness is non-trivial.

In this study, we develop a collaborative filtering method with \emph{individual} item fairness for large-scale applications with many users and a large item catalogue.
We take an in-processing approach and build a method based on iALS to inherit its advantages in ranking quality and scalability. 
The downside of such an approach is that there is no guarantee of the properties of interest in conventional studies, e.g. envy-freeness~\cite{patro2020fairrec,do2022online,saito2022fair} and Lorenz efficiency~\cite{do2021two}.
This study explores the possibility of developing a practical algorithm by trading theoretical properties/guarantees for feasibility.

\section{Proposed Method}
\subsection{Problem Setting}
To realise a scalable and fairness-aware method, 
we consider the minimisation problem of the iALS loss $L(\bfV, \bfU)$ with fairness regularisation as follows:
\begin{align}
    \min_{\bfV, \bfU} L(\bfV, \bfU) + \lambda_f R_f(\bfV, \bfU),
    \label{eq:obj-imf-fair}
\end{align}
where $R_f(\bfV, \bfU)$ is a regularisation term to induce item fairness, and $\lambda_f$ is the weight parameter to balance ranking quality and item fairness.
As we discussed above, the scalability of iALS relies on the simplicity of the objective.
To retain this desirable property,
we develop a tractable fairness regulariser $R_f$.

Measures of unfairness and inequality are often based on \emph{variability}.
The Gini index (or Gini mean difference) is a widely utilised measure of inequality and is defined as follows:
\begin{align}
  \label{eq:gini-index}
  \text{Gini}(\bo) = \frac{1}{2\norm{\bo}_1|\calV|^2}\sum_{j \in \calV}\sum_{l \in \calV}|o_j - o_l|,
\end{align}
where $\bo \in \R^{|\calV|}$ is an $|\calV|$-dimensional vector, whose $j$-th element $o_j$ indicates the utility (e.g. exposure) of item $j$.
In contrast to non-differentiable Gini indices, various measures of inequality have been explored, e.g. the standard deviation~\cite{do2021two}.

Optimising fairness-agnostic item rankings also introduce dependency between items in the optimisation for each (independent) user, as it involves the relative order between items.
iALS optimises item rankings in the space of preference scores rather than that of rankings or a probability simplex for efficiency; this is one aspect of its optimisation separability with respect to items.
Based on these virtues of iALS, we design a surrogate measure of exposure inequality based on the variability of the predicted item merit (i.e. the average preference of each item).
Denoting the predicted score for user $i$ and item $j$ by $\hat{r}_{i,j}=(\bfU\bfV^\top)_{i,j}$,
we consider the squared L2 norm of the predicted item merit as a measure of inequality:
\begin{align*}
  R_f(\bfV, \bfU)
  &= \frac{1}{2}\sum_{j \in \calV}\left(\frac{1}{|\calU|}\sum_{i \in \calU}\hat{r}_{i,j}\right)^2 
  = \frac{1}{2}\norm{\frac{1}{|\calU|}\bfV\bfU^\top\1}_2^2,
\end{align*}
where $\1$ is the $|\calU| \times 1$ column vector of which the elements are all 1.
Observe that $R_f$ takes a large value for items of which the average predicted scores are either extremely large or small.
It is differentiable and hence easier to optimise than non-differentiable measures such as the Gini index.
However, unfortunately, optimising $R_f$ is still not straightforward in large-scale settings because it destroys the optimisation separability with respect to the rows of $\bfU$ owing to the average user factor $(1/|\calU|)\bfU^\top\1$.

\subsection{Algorithm}

\subsubsection{Alternating Direction Method of Multipliers}
To enable parallel optimisation with respect to users,
we adopt an approach based on ADMM,
which is an optimisation framework with high parallelism~\cite{boyd2011distributed} and has been adopted for scalable recommender systems~\cite{yu2014distributed,cheng2014lorslim,smith2017constrained,ioannidis2019coupled,steck2020admm,steck2021negative}.
To decouple the row-wise and column-wise dependencies in $\bfU$ introduced by the fairness regulariser $(1/2)\|(1/|\calU|)\bfV\bfU^\top\1\|_2^2$,
we first reformulate the optimisation problem by introducing an auxiliary variable $\s \in \R^d$ as follows:
\begin{align}
\label{eq:fiadmm-obj}
  \min_{\bfV,\bfU,\s} &L(\bfV,\! \bfU) + \frac{\lambda_f}{2}\norm{\bfV\s}_2^2,\,\,\,\, \text{s.t.~} \s = \frac{1}{|\calU|}\bfU^\top\1.  
\end{align}
Here, we replaced $(1/|\calU|)\bfU^\top\1$ in the fairness regulariser with $\s$ while introducing an additional linear equality constraint.

Following ADMM, this can be further reformulated to the following saddle-point optimisation:
\begin{align*}
  \min_{\bfV,\bfU,\s} \max_{\w}L_\rho(\bfV, \bfU, \s, \w),
\end{align*}
where
\begin{align*}
L_\rho(\bfV, \bfU, \s, \w) \!&=\! L(\bfV, \bfU) \!+\! \frac{\lambda_f}{2}\norm{\bfV\s}_2^2 \!+\! \frac{\rho}{2}\norm{\frac{1}{|\calU|}\bfU^\top\1 \!-\! \s \!+\! \w}_2^2 \!-\! \frac{\rho}{2}\|\w\|_2^2.
\end{align*}
Here, $L_\rho$ is the Lagrangian augmented by the penalty term with weight $\rho>0$, and $\w \in \R^{d}$ is the dual variable (i.e. Lagrange multipliers) scaled by $1/\rho$.
Optimisation in the $(k+1)$-th step is performed by alternately updating each variable as follows:
\begin{align*}
  \bfV^{k+1} &= \argmin_\bfV  L_\rho(\bfV, \bfU^{k}, \s^{k}, \w^{k}), \\
  \bfU^{k+1} &= \argmin_\bfU  L_\rho(\bfV^{k+1}, \bfU, \s^{k}, \w^{k}), \\
  \s^{k+1} &= \argmin_\s  L_\rho(\bfV^{k+1}, \bfU^{k+1}, \s, \w^{k}), \\
  \w^{k+1} &= \w^{k} + \frac{1}{|\calU|}(\bfU^{k+1})^\top\1 - \s^{k+1},
\end{align*}
The update of $\w$ corresponds to the gradient ascent with respect to the dual problem $\max_{\w}\min_{\bfU,\bfV,\s}L_\rho(\bfV,\bfU,\s,\w)$ with step size $\rho$~\cite{boyd2011distributed}.

\subsubsection{Update of $\bfV$}
We derive the update of $\bfV$ in the $(k+1)$-th step, which comprises independent optimisation problems with respect to the rows of $\bfV$.
Suppose that $\bv_j^{k+1} \in \R^{d}$ and $\br_{:, j} \in \{0,1\}^{|\calU|}$ are the column vectors indicating the $j$-th row of $\bfV^{k+1}$ and the $j$-th column of $\bR$, respectively.
The update is then performed by solving the following linear system:
\begin{align*}
   \bv_{j}^{k+1} &=\argmin_{\bv_j}  \Bigl\{\frac{1}{2}\norm{\br_{:, j} \odot (\br_{:, j} - \bfU^{k}\bv_j)}_2^2 + \frac{\alpha_0}{2}\norm{\bfU^{k}\bv_j}_2^2 \\
   &\,\,\,\,\,\,\,\,\,\,\,\,\,\,\,\,\,\,\,\,\,+ \frac{\lambda_V^{(j)}}{2}\norm{\bv_j}_2^2 + \frac{\lambda_f}{2}\left(\bv_j^\top \s^{k}\right)^2\Bigr\}\\
  &= \left(\sum_{i \in \calU}r_{i,j} \bu_i^{k} ({\bu_i^{k}})^\top + \alpha_0\bfG_U^{k} + \lambda_f \s^{k}(\s^{k})^\top + \lambda_V^{(j)}\I\right)^{-1}\sum_{i \in \calU}r_{i,j}\bu_i^{k},
\end{align*}
where $\bfG_U^{k}=\sum_{i \in \calU} \bu_i^{k}{\bu_i^{k}}^\top$ is the Gramian of the user factors in the $k$-th step.
Notably, we can pre-compute $\bfG_U^{k}$ and $\s^{k}(\s^{k})^\top$, and the update of $\bfV$ achieves the same complexity as that of iALS.

\subsubsection{Update of $\bfU$}
Updating $\bfU$ is the most intricate part of our algorithm.
At the $(k\!+\!1)$-th step, our aim is to solve the following problem:
\begin{align*}
   \bfU^{k+1} \!=\! \argmin_{\bfU} &\Biggl\{ L(\bfV^{k+1}\!, \bfU) \!+\! \frac{\lambda_f}{2}\|\bfV^{k+1}\s^k\|_2^2 \!+\! \frac{\rho}{2}\norm{\frac{1}{|\calU|}\bfU^\top\1 \!-\! \s^{k} \!+\! \w^{k}}_2^2\Biggr\}.
\end{align*}
Observe that the penalty term of ADMM (the fourth term of RHS) destroys the independence between the rows of $\bfU$.
We resolve this using a proximal gradient method (a.k.a. forward-backward splitting)~\cite{rockafellar1976monotone,duchi2009efficient,liu2019linearized}.
We consider a linear approximation (i.e. the first-order Taylor expansion around the current estimate $\bfU^{k}$) of the objective except for the ADMM penalisation.
This yields the following approximated objective:
\begin{align*}
   \bfU^{k+1} = \argmin_{\bfU} &\Biggl\{ \langle \bfU - \bfU^{k}, \nabla_{\bfU}g(\bfV^{k+1}, \bfU^{k}, \s^{k}) \rangle_F + \frac{1}{2\gamma}\norm{\bfU - \bfU^{k}}_F^2 \\
   &+ \frac{\rho}{2}\norm{\frac{1}{|\calU|}\bfU^\top\1 - \s^{k} + \w^{k}}_2^2\Biggr\}
\end{align*}
where
\begin{align*}
   g(\bfV,\bfU,\s) &= \frac{1}{2}\norm{\bR \odot (\bR - \bfU\bfV^\top)}_F^2 + \frac{\alpha_0}{2}\norm{\bfU\bfV^\top}_F^2  + \frac{\lambda_f}{2}\norm{\bfV\s}_2^2 \\
   &+ \frac{1}{2}\norm{\bLambda_U^{1/2} \bfU}_F^2 + \frac{1}{2}\norm{\bLambda_V^{1/2} \bfV}_F^2.
\end{align*}
We here introduce a regularisation term $(1/2\gamma)\lVert\bfU - \bfU^{k}\rVert_F^2$, which is referred to as the proximal term~\cite{rockafellar1976monotone}.
By completing the square, the above objective can be rearranged into the following parallel and non-parallel computing steps: 
\begin{align*}
   \bfU^{k+1} &= \argmin_{\bfU} \frac{\rho}{2}\norm{\frac{1}{|\calU|}\bfU^\top\1 \!-\! \s^{k} \!+\! \w^{k}}_2^2 + \frac{1}{2\gamma}\norm{\bfU \!-\! \left(\bfU^{k}\!-\!\gamma\nabla_{\bfU}g^{k}\right)}_F^2 \\
   & = \underbrace{\mathrm{prox}_{\gamma}^{k}}_{\text{non-parallel}}\!\!( \underbrace{\vphantom{\mathrm{prox}_{\gamma}^{k}}\bfU^{k}-\gamma\nabla_{\bfU}g^{k}}_{\text{parallel}}),
\end{align*}
where
\begin{align*}
    \mathrm{prox}_{\gamma}^{k}(\widetilde{\bfU}) &= \argmin_{\bfU} \frac{\rho}{2}\norm{\frac{1}{|\calU|}\bfU^\top\1 - \s^k + \w^k}_2^2 + \frac{1}{2\gamma}\norm{\bfU-\widetilde{\bfU}}_F^2\\
    &= \left(\frac{\rho}{|\calU|^2} \1\1^\top + \frac{1}{\gamma}\I\right)^{-1}\left(\frac{1}{\gamma}\widetilde{\bfU} + \frac{\rho}{|\calU|}\1(\s^{k} - \w^{k})^\top\right).
\end{align*}
Here, $\nabla_{\bfU}g^{k}$ is used to represent $\nabla_{\bfU}g(\bfV^{k+1}, \bfU^{k}, \s^{k})$.
Notice here that $\bfU^{k}-\gamma\nabla_{\bfU}g^{k}$ corresponds to a gradient descent of the iALS objective\footnote{Note that $\nabla_{\bfU}g(\bfV,\bfU,\s)$ is equivalent to the derivative of the iALS objective with respect to $\bfU$ because we can ignore the constant fairness regulariser $(\lambda_f/2)\|\bfV\s\|_2^2$.} with a learning rate $\gamma$.
Therefore, $\bfU$ can be updated in two row-wise parallel and non-parallel steps, that is, (1) gradient descent $\widetilde{\bfU}^{k+1}=\bfU^{k}-\gamma\nabla_{\bfU}g^{k}$ and (2) proximal mapping $\bfU^{k+1}=\mathrm{prox}_{\gamma}^{k}(\widetilde{\bfU}^{k+1})$.

\noindent\textbf{Parallel Gradient Computation.}
The gradient $\nabla_{\bfU}g(\bfV^{k+1}, \bfU^{k}, \s^{k})$ can be independently computed for each row of $\bfU$ as follows:
\begin{align*}
   \nabla_{\bu_i}g^{k} 
   &= \left(\sum_{j \in \calV} r_{i,j} {\bv_j}^{k+1}{\bv_j^{k+1}}^\top + \alpha_0\bfG_V^{k+1} + \lambda_U^{(i)} \I\right)\bu_i^{k} - \sum_{j \in \calV}r_{i,j}\bv_j^{k+1}
\end{align*}
Similar to iALS, we can efficiently compute the gradient by pre-computing the Gramian $\bfG_V^{k+1}=\sum_{j \in \calV} {\bv_j^{k+1}}(\bv_j^{k+1})^\top$.
Therefore, the gradient descent $\bfU^k - \nabla_{\bfU}g^k$ can be performed in parallel with respect to users.
Notably, we can avoid the computation of the inverse Hessian in $\calO(d^3)$ unlike the $\bfU$ step of iALS.

\noindent\textbf{Efficient Proximal Mapping.}
The proximal mapping step requires inversion of the $|\calU| \times |\calU|$ matrix, the computational complexity of which is $\calO(|\calU|^3)$ for a na\"ive computation.
This is problematic because, in practise, $\rho$ and $\gamma$ may increase/decrease during the iterations~\cite{boyd2011distributed}.
However, we can efficiently compute an inverse matrix.
The Sherman-Morrison formula~\cite{sherman1950adjustment} (a special case of the Woodbury matrix identity~\cite{woodbury1950inverting}) yields the following matrix inverse:
\begin{align*}
   \left(\frac{\rho}{|\calU|^2}\1\1^\top + \frac{1}{\gamma}\I\right)^{-1} &= 
   -\frac{(\gamma\I)(\rho/|\calU|^2)\1\1^\top(\gamma\I)}{1 + (\rho/|\calU|^2)\1^\top (\gamma \I) \1} + \gamma \I \\
   &= \gamma\left(-\frac{\rho}{|\calU|^2\left(\frac{1}{\gamma} + \frac{\rho}{|\calU|}\right)}\1\1^\top + \I\right).
\end{align*}
Therefore, the proximal mapping $\mathrm{prox}_{\gamma}^k$ can be obtained as the following closed-form solution:
\begin{align*}
\mathrm{prox}_{\gamma}^{k}(\bfU) 
&= \left(\frac{-\rho}{|\calU|^2\left(\frac{\rho}{|\calU|} + \frac{1}{\gamma}\right)}\1\1^\top+\I \right)\left(\bfU + \frac{\rho\gamma}{|\calU|}\1(\s^{k} - \w^{k})^\top\right).
\end{align*}
The na\"ive computation of $\mathrm{prox}_{\gamma}^k$ is still computationally costly owing to the multiplication of $|\calU|\times|\calU|$ and $|\calU| \times d$ matrices in $\calO(|\calU|^2d)$.
However, this matrix multiplication can be efficiently performed by (1) computing $\widehat{\bfU}^{k+1}=\widetilde{\bfU}^{k+1} + \frac{\rho\gamma}{|\calU|}\1(\s^k - \w^k)^\top$ in parallel with respect to each row of $\widehat{\bfU}^{k+1}$, (2) summing up the rows of $\widehat{\bfU}^{k+1}$ by $\widehat{\bu}^{k+1} = (\widehat{\bfU}^{k+1})^\top\1$, and then (3) adding $c \cdot \widehat{\bu}^{k+1}$ to each row of $\widehat{\bfU}^{k+1}$ where $c=-\rho(|U^2|(\nicefrac{\rho}{|\calU|}+\nicefrac{1}{\gamma}))^{-1}$.
Thus, the computational cost of this step is $\calO(|\calU|d)$, which is more efficient than $\calO(|\calU|^2d)$ of na\"ive multiplication.
This computational efficiency is advantageous even when $\rho$ and $\gamma$ are fixed during optimisation.

\subsubsection{Update of $\s$.}
The update of $\s$ is performed by computing the closed-form solution to the following problem:
\begin{align*}
   \s^{k+1} &= \argmin_{\s} \left\{ \frac{\lambda_f}{2}\norm{\bfV^{k+1}\s}_2^2 + \frac{\rho}{2}\norm{\frac{1}{|\calU|}(\bfU^{k+1})^\top\1 - \s + \w^{k}}_2^2\right\} \\
   &= \rho\left(\lambda_f \bfG_V^{k+1} + \rho\I\right)^{-1}\left(\frac{1}{|\calU|}(\bfU^{k+1})^\top \1 + \w^{k}\right).
\end{align*}
The Gramian $\bfG_V^{k+1}$ can be reused for this step following the pre-computation in the $\bfU$ step.
The cost of this step is thus $\calO(|\calU|d + d^3)$, which includes (1) the computation of $(1/|\calU|)(\bfU^{k+1})^\top \1$ and (2) the solution of a linear system of size $d^2$.

\begin{algorithm}[th]
  \caption{Fair Implicit ADMM}
  \small
  \label{alg:fiadmm}
  \begin{algorithmic}[1]
    \REQUIRE{Feedback matrix $\bR$}
    \STATE $\bu_{i}^{(0)} \sim \mathcal{N}(0, (\sigma/\sqrt{d})\I)$ for $\forall i \in \calU$
    \STATE $\bv_{j}^{(0)} \sim \mathcal{N}(0, (\sigma/\sqrt{d})\I)$ for $\forall j \in \calV$
    \STATE $\s^{(0)} \gets (1/|\calU|)(\bfU^{(0)})^\top\1$, $\w^{(0)} \gets \vec{0}$
    \FOR{$k=0,\dots,T-1$} 
      \STATE $\bfG_U^{k} \gets \sum_{i \in \calU} \bu_i^{k}{\bu_i^{k}}^\top$  \hfill // \textit{$\calO(|\calU|d^2)$}
      \STATE $\bfG_s^{k} \gets \s^{k}{\s^{k}}^\top$  \hfill // \textit{$\calO(d^2)$}
      \FOR{$j=1,\dots,|\calV|$} \hfill // \textit{\textbf{parallelisable loop}}
        \STATE $\bfG_{j}^{k} \gets \sum_{i \in \calU} r_{i,j} \bu_i^{k}{\bu_i^{k}}^\top$ \hfill // \textit{$\calO((\mathrm{nz}(\bR)/|\calV|)d^2)$}
        \STATE $\bv_j^{k+1} \gets \left(\bfG_j^{k} +\alpha_0\bfG_U^{k} + \lambda_f \bfG_s^{k} + \lambda_V^{(j)}\I\right)^{-1} \sum_{i \in \calU} r_{i,j}\bu_i^{k}$ \hfill// \textit{$\calO(d^3)$}
      \ENDFOR
      \STATE $\bfG_V^{k+1} \gets \sum_{j \in \calV} \bv_j^{k+1}{\bv_j^{k+1}}^\top$  \hfill // \textit{$\calO(|\calV|d^2)$}
      \FOR{$i=1,\dots,|\calU|$} \hfill // \textit{\textbf{parallelisable loop}}
        \STATE $\bfG_i^{k+1} \gets \sum_{j \in \calV} r_{i,j} \bv_j^{k+1}{\bv_j^{k+1}}^\top$ \hfill // \textit{$\calO((\mathrm{nz}(\bR)/|\calU|)d^2)$}
        \STATE $\nabla_{\bu_i}g^{k+1} \gets \left(\bfG_i^{k+1} \!+\! \alpha_0\bfG_V^{k+1} \!+\! \lambda_U^{(i)} \I\right)\bu_i^{k} - \sum_{j \in \calV}r_{i,j}\bv_j^{k+1}$ \hfill // \textit{$\calO(d^2)$}
        \STATE $\widetilde{\bu}_{i}^{k+1} \gets \bu_{i}^{k} - \gamma\nabla_{\bu_i}g^{k+1}$ \hfill // \textit{$\calO(d)$}
        \STATE $\widetilde{\bu}_{i}^{k+1} \gets \widetilde{\bu}_{i}^{k+1} + \frac{\rho\gamma}{|\calU|}(\s^{k} - \w^{k})$ \hfill // \textit{$\calO(d)$}
      \ENDFOR
      \STATE $\widehat{\bu}^{k+1} \gets \sum_{i \in \calU}\widetilde{\bu}_{i}^{k+1}$ \hfill // \textit{$\calO(|\calU|d)$}
      \FOR{$i=1,\dots,|\calU|$} \hfill // \textit{\textbf{parallelisable loop}}
        \STATE $\bu_{i}^{k+1} \gets \widetilde{\bu}_{i}^{k+1}  -\rho(|U^2|(\nicefrac{\rho}{|\calU|}+\nicefrac{1}{\gamma}))^{-1}\widehat{\bu}^{k+1}$ \hfill // \textit{$\calO(d)$}
      \ENDFOR
      \STATE $\bt^{k+1} \gets \frac{1}{|\calU|}(\bfU^{k+1})^\top \1$ \hfill // \textit{$\calO(|\calU|d)$}
      \STATE $\s^{k+1} \gets \rho\left(\lambda_f \bfG_V^{k+1} + \rho\I\right)^{-1}\left(\bt^{k+1} - \w^{k}\right)$ \hfill // \textit{$\calO(d^3)$}
      \STATE $\w^{k+1} \gets \w^{k} + \bt^{k+1} - \s^{k+1}$ \hfill // \textit{$\calO(d)$}
    \ENDFOR
    \RETURN $\bfU^{\top}, \bfV^{\top}$ 
  \end{algorithmic}
\end{algorithm}

\subsection{Complexity Analysis}
\label{section:complexity-analysis}
\cref{alg:fiadmm} shows the detailed implementation of the proposed algorithm, \emph{fair implicit ADMM (fiADMM)}.
First, the user/item factors are initialised with independent normal noise with a $\sigma/\sqrt{d}$ standard deviation~\cite{rendle2021revisiting}.
In line 10 of \cref{alg:fiadmm}, we pre-compute $\bfG_V^{k}$, which can be reused in the update of $\bfU$ and $\s$.
The calculated average user vector $\bt = (1/|\calU|)(\bfU)^\top \1$ can be reused for both the $\s$ and $\w$ steps; hence, we compute this in line 21 in Algorithm~\ref{alg:fiadmm}.
Consequently, the computational costs for updating $\bfV$, $\bfU$, $\s$, $\w$ are, respectively,
(1) $\calO(\mathrm{nz}(\bR)d^2 \!+\! |\calV|d^3)$,
(2) $\calO(\mathrm{nz}(\bR)d^2 \!+\! |\calU|d^2)$,
(3) $\calO(|\calU|d \!+\! d^3)$,
and (4) $\calO(d)$.
Therefore, the overall cost is $\calO(\mathrm{nz}(\bR)d^2 \!+\! |\calU|d^2 \!+\! |\calV|d^3)$, which is faster than $\calO(\mathrm{nz}(\bR)d^2 \!+\! (|\calU| \!+\! |\calV|)d^3)$ of iALS; this is because we can avoid solving the linear system when updating $\bfU$ owing to the proximal gradient method with the efficient $\mathrm{prox}_{\gamma}^{k}$.
In exchange for improved runtime per step and scalability, our algorithm would slow down the convergence compared with iALS because of the linear approximation when updating $\bfU$.

\subsection{Convergence Analysis}
The objective defined in \cref{eq:fiadmm-obj} has more than two variables (i.e. three-block optimisation), and the variables are \emph{coupled} (e.g. $\bfU, \bfV$ in the iALS loss function).
Multi-block ADMM does not guarantee convergence in general~\cite{chen2016direct}. 
Various algorithms have been developed for optimisation separability and provable convergence under coupled variables~\cite{wang2014parallel,deng2017parallel,liu2019linearized}.
\citet{liu2019linearized} proposed a variant of ADMM for non-convex problems, which completely decouples variables by introducing linear approximation when updating \emph{all} coupled ones, thereby enabling parallel gradient descent.
By contrast, fiADMM applies linearisation only to the $\bfU$ step and carries out the update steps alternately.
This strategy enables second-order acceleration in the update of $\bfV$ and $\s$; however, it might impair convergence.
Considering this, we provide a convergence guarantee for fiADMM as the following theorem.
\begin{theorem}
\label{thm:fiADMM}
Assume that there exist constants $C_V, C_U, C_{\s} > 0$ such that $\|\bfV^{k}\|_F^2\leq C_V$, $\|\bfU^{k}\|_F^2\leq C_U$, $\|\s^{k}\|_2^2\leq C_{s}$ for $\forall k\geq 0$.
For $\rho \!\geq\! \max\left(\frac{24\lambda_f^2 C_V C_{\s}}{\underline{\lambda}_V}, \frac{1}{2}\!+\!\sqrt{\frac{1}{4}\!+\!6\lambda_f^2 C_V^2}\right)$ and $\gamma \!\leq\! \frac{1}{\sqrt{|\calV|}((1+\alpha_0)C_V+\bar{\lambda}_U) + 1}$, where $\bar{\lambda}_U=\max_{i \in \calU}\lambda_U^{(i)}$ and $\underline{\lambda}_V=\min_{j \in \calV}\lambda_V^{(j)}$, the augmented Lagrangian $L_{\rho}(\bfV^{k},\bfU^{k},\s^{k},\w^{k})$ converges to some value, while residual norms $\|\bfV^{k+1}\!-\bfV^{k}\|_F,\|\bfU^{k+1}\!-\bfU^{k}\|_F,\|\s^{k+1}\!-\s^{k}\|_2$, and $\|\w^{k+1}\!-\w^{k}\|_2$ converge to $0$.
Furthermore, the gradients of $L_{\rho}$ with respect to $\bfV$, $\bfU$, $\s$, and $\w$ converge to $0$.
\end{theorem}
\begin{proof}[Proof of Theorem \ref{thm:fiADMM}]
In the proof, we use the following lemma on the smoothness of $g$:
\begin{lemma}
\label{lem:g_smooth}
For any $\bfV,\bfV'\in \mathbb{R}^{|\calV|d}$, $\s,\s'\in \mathbb{R}^d$, and $\bfU,\bfU'\in \mathbb{R}^{|\calU|d}$, function $g$ satisfies the following inequalities:
\begin{align*}
    &\|\nabla_{\bfV}g(\bfV,\!\bfU,\!\s) \!-\! \nabla_{\bfV}g(\bfV',\!\bfU,\!\s)\|_F \!\leq\! \sqrt{|\calV|}\!\left(\!\left(\!1 \!+\! \alpha_0\!\right)\!\|\bfU\|_F^2 \!+\! \|\s\|_2^2 \!+\! \bar{\lambda}_V\!\right)\!\|\bfV\!-\!\bfV'\|_F ,\\
    &\|\nabla_{\bfU} g(\bfV,\! \bfU,\! \s) \!-\! \nabla_{\bfU} g(\bfV,\! \bfU',\! \s')\|_F \!\leq\! \sqrt{|\calU|}\!\left(\!(\!1 \!+\! \alpha_0)\|\bfV\|_F^2 + \bar{\lambda}_U\right)\|\bfU\!-\!\bfU'\|_F ,\\
    &\|\nabla_{\s}g(\bfV,\! \bfU,\! \s) \!-\! \nabla_{\s}g(\bfV,\! \bfU',\! \s')\|_2 \!\leq\! \lambda_f\|\bfV\|_F^2\|\s\!-\!\s'\|_2 ,\\
    &\|\nabla_{\s}g(\bfV,\! \bfU,\! \s) \!-\! \nabla_{\s}g(\bfV',\! \bfU,\! \s)\|_2 \!\leq\! \lambda_f(\|\bfV\|_F + \|\bfV'\|_F)\|\s\|_2\|\bfV\!-\!\bfV'\|_F,
\end{align*}
where $\bar{\lambda}_U = \max_{i \in \calU}\lambda_U^{(i)}$ and $\bar{\lambda}_V = \max_{j \in \calV}\lambda_V^{(j)}$.
\end{lemma}

We prove the first part of the theorem.
We decompose the difference of $L_\rho$ before and after a single epoch update into that before and after each alternating step.
\begin{align}
\label{eq:decomposition_lagrangian}
    \nonumber
    &L_\rho(\bfV^{k+1}\!, \bfU^{k+1}\!, \s^{k+1}\!, \w^{k+1}) - L_\rho(\bfV^{k}\!, \bfU^{k}\!, \s^{k}\!, \w^{k}) \\
    &= \left(L_\rho(\bfV^{k+1}\!, \bfU^{k+1}\!, \s^{k}\!, \w^{k}) - L_\rho(\bfV^{k}\!, \bfU^{k}\!, \s^{k}\!, \w^{k})\right) \nonumber\\
    &\phantom{=} + \left(L_\rho(\bfV^{k+1}\!, \bfU^{k+1}\!, \s^{k+1}\!, \w^{k}) - L_\rho(\bfV^{k+1}\!, \bfU^{k+1}\!, \s^{k}\!, \w^{k})\right) \nonumber\\
    &\phantom{=} + \left(L_\rho(\bfV^{k+1}\!, \bfU^{k+1}\!, \s^{k+1}\!, \w^{k+1}) - L_\rho(\bfV^{k+1}\!, \bfU^{k+1}\!, \s^{k+1}\!, \w^{k})\right).
\end{align}
By Lemma \ref{lem:g_smooth}, we obtain the upper bound on each term in the RHS:
\begin{lemma}
\label{lem:vu_sub_ub}
The update of $\bfV$ and $\bfU$ in the $(k+1)$-step satisfies
\begin{align*}
    &L_\rho(\bfV^{k+1}, \bfU^{k+1}, \s^{k}, \w^{k}) - L_\rho(\bfV^{k}, \bfU^{k}, \s^{k}, \w^{k}) \\
    &\leq \frac{\sqrt{|\calU|}((1 + \alpha_0)C_V + \bar{\lambda}_U) - 1/\gamma}{2}\|\bfU^{k+1}\!-\!\bfU^{k}\|_F^2 - \frac{\underline{\lambda}_V}{2}\|\bfV^{k+1}\!-\!\bfV^{k}\|_F^2.
\end{align*}
\end{lemma}
\begin{lemma}
\label{lem:s_sub_ub}
The update of $\s$ in the $(k+1)$-th step satisfies
\begin{align*}
  L_\rho(\bfV^{k+1}\!, \bfU^{k+1}\!, \s^{k+1}\!, \w^{k}) \!-\! L_\rho(\bfV^{k+1}\!, \bfU^{k+1}\!, \s^{k}\!, \w^{k})
  \!\leq\! - \frac{\rho}{2}\|\s^{k+1}\!-\!\s^{k}\|_2^2.
\end{align*}
\end{lemma}
\begin{lemma}
\label{lem:w_sub_ub}
The update of $\w$ in the $(k+1)$-th step satisfies
\begin{align*}
    &L_\rho(\bfV^{k+1}, \bfU^{k+1}, \s^{k+1}, \w^{k+1}) - L_\rho(\bfV^{k+1}, \bfU^{k+1}, \s^{k+1}, \w^{k}) \\
    &\leq \frac{3\lambda_f^2C_V^2}{\rho}\|\s^{k+1}-\s^{k}\|_2^2+\frac{6\lambda_f^2 C_V C_s}{\rho}\|\bfV^{k+1}-\bfV^{k}\|_F^2.
\end{align*}
\end{lemma}
By combining \cref{eq:decomposition_lagrangian} and Lemmas \ref{lem:vu_sub_ub} to \ref{lem:w_sub_ub}, under the assumptions about $\rho$ and $\gamma$, we have
\begin{align}
\label{eq:L_rho_sub_ub}
    &L_\rho(\bfV^{k+1}, \bfU^{k+1}, \s^{k+1}, \w^{k+1}) - L_\rho(\bfV^{k}, \bfU^{k}, \s^{k}, \w^{k}) \nonumber\\
    &\leq \frac{\sqrt{|\calU|}((1 + \alpha_0)C_V + \bar{\lambda}_U) - 1/\gamma}{2}\|\bfU^{k+1}-\bfU^{k}\|_F^2 \nonumber\\
    &\phantom{\leq} + \left(- \frac{\underline{\lambda}_V}{2} \!+\! \frac{6\lambda_f^2 C_V C_s}{\rho}\right)\|\bfV^{k+1} \!-\! \bfV^{k}\|_F^2 
    + \left(- \frac{\rho}{2} \!+\! \frac{3\lambda_f^2 C_V^2}{\rho}\right)\|\s^{k+1}-\s^{k}\|_2^2 \nonumber\\
    & \leq -\frac{1}{2}\|\bfU^{k+1}\!-\!\bfU^{k}\|_F^2 - \frac{\underline{\lambda}_V}{4}\|\bfV^{k+1}\!-\!\bfV^{k}\|_F^2 
    - \frac{1}{2}\|\s^{k+1}\!-\!\s^{k}\|_2^2 \leq 0.
\end{align}
Therefore, $L_{\rho}(\bfV^{k}, \bfU^{k}, \s^{k}, \w^{k})$ is monotonically decreasing.

Here, we obtain the following lower bound on $L_{\rho}(\bfV^{k}, \bfU^{k}, \s^{k}, \w^{k})$:
\begin{lemma}
\label{lem:L_rho_lb}
$\bfV^{k}, \bfU^{k}, \s^{k}$, and $\w^{k}$ updated by fiADMM satisfy
\begin{align*}
    L_\rho(\bfV^{k}, \bfU^{k}, \s^{k}, \w^{k}) \geq \frac{\rho-\lambda_f C_V}{2}\||\calU|^{-1}(\bfU^{k})^\top \1-\s^{k}\|_2^2.
\end{align*}
\end{lemma}
Thus, when $\rho \geq \lambda_f C_V$ holds, $L_\rho(\bfV^{k}, \bfU^{k}, \s^{k}, \w^{k})$ is lower bounded by $0$.
Therefore, owing to its monotonic decrease, $L_\rho(\bfV^{k}, \bfU^{k}, \s^{k}, \w^{k})$ converges to some constant value, and the LHS of \cref{eq:decomposition_lagrangian} converges to $0$.
From \cref{eq:L_rho_sub_ub} and the fact that $L_\rho(\bfV^{k+1}, \bfU^{k+1}, \s^{k+1}, \w^{k+1}) - L_\rho(\bfV^{k}, \bfU^{k}, \s^{k}, \w^{k})$ converges to $0$, $\|\bfV^{k+1}-\bfV^{k}\|_F$, $\|\bfU^{k+1}-\bfU^{k}\|_F$, and $\|\s^{k+1}-\s^{k}\|_2$ also converge to $0$.
Because it holds that
\begin{align*}
    &L_\rho(\bfV^{k+1}, \bfU^{k+1}, \s^{k+1}, \w^{k+1}) - L_\rho(\bfV^{k+1}, \bfU^{k+1}, \s^{k+1}, \w^{k}) \\
    &= \rho\|\w^{k+1}-\w^{k}\|_2^2 \leq  \frac{3\lambda_f^2C_V^2}{\rho}\|\s^{k+1}-\s^{k}\|_2^2+\frac{6\lambda_f^2C_VC_s}{\rho}\|\bfV^{k+1}-\bfV^{k}\|_F^2,
\end{align*}
we can also state that $\|\w^{k+1}-\w^{k}\|_2$ converges to $0$.

We next prove the second part of the theorem.
Since $\bfV^{k+1}$ minimises $L_{\rho}(\bfV,\bfU^{k},\s^{k},\w^{k})$, it holds that $\nabla_{\bfV}L_{\rho}(\bfV^{k+1},\bfU^{k},\s^{k},\w^{k})=0$, and we obtain the following inequality from Lemma \ref{lem:g_smooth}:
\begin{flalign*}
    &\|\nabla_{\bfV}L_{\rho}(\bfV^{k}\!,\bfU^{k}\!,\s^{k}\!,\w^{k})\|_F &\\
    &= \|\nabla_{\bfV}L_{\rho}(\bfV^{k+1}\!,\bfU^{k}\!,\s^{k}\!,\w^{k}) - \nabla_{\bfV}L_{\rho}(\bfV^{k}\!,\bfU^{k}\!,\s^{k}\!,\w^{k})\|_F \\
    &= \|\nabla_{\bfV}g(\bfV^{k+1}\!,\bfU^{k}\!,\s^{k}) - \nabla_{\bfV}g(\bfV^{k}\!,\bfU^{k}\!,\s^{k})\|_F \\
    &\leq \sqrt{|\calV|}\left(\left(1 \!+\! \alpha_0\right)C_U \!+\! C_s \!+\! \bar{\lambda}_V\right)\|\bfV^{k+1}-\bfV^{k}\|_F.
\end{flalign*}
Since $\|\bfV^{k+1}-\bfV^{k}\|_F$ converges to $0$, $\|\nabla_{\bfV}L_{\rho}(\bfV^{k},\bfU^{k},\s^{k},\w^{k})\|_F$ and then $\nabla_{\bfV}L_{\rho}(\bfV^{k},\bfU^{k},\s^{k},\w^{k})$ converge to $0$.
Similarly, we have:
\begin{align*}
    &\|\nabla_{\bfU}L_{\rho}(\bfV^{k},\bfU^{k},\s^{k},\w^{k})\|_F \\
    &= \|\nabla_{\bfU}g(\bfV^{k},\bfU^{k},\s^{k}) + \rho |\calU|^{-\!2}\1\1^\top \bfU^{k} + \rho |\calU|^{-1}\1(\w^{k}\!-\!\s^{k})^\top\|_F \\
    &= \|\nabla_{\bfU}g(\bfV^{k}\!,\bfU^{k}\!,\s^{k}) \!-\! \nabla_{\bfU}g(\bfV^{k}\!,\bfU^{k\!-\!1}\!, \s^{k\!-\!1}) + \rho |\calU|^{-1}\1(\w^{k}\!-\!\w^{k\!-\!1})^{\!\top} \\
    &\phantom{=} - \rho |\calU|^{-\!1}\1(\s^{k}\!-\!\s^{k\!-\!1})^{\!\top}  \!-\! (1/\gamma)(\bfU^{k}\!-\!\bfU^{k\!-\!1})\|_F \\
    &\leq \|\nabla_{\bfU}g(\bfV^{k}\!,\bfU^{k}\!,\s^{k}) \!-\! \nabla_{\bfU}g(\bfV^{k}\!,\bfU^{k\!-\!1}\!, \s^{k\!-\!1})\|_F  \!+\! (1/\gamma)\|\bfU^{k}\!-\!\bfU^{k\!-\!1}\|_F \\
    &\phantom{=}  \!+\! \rho|\calU|^{-\!1} (\|\1(\w^{k}\!-\!\w^{k\!-\!1})^{\!\top}\|_F \!+\! \|\1(\s^{k}\!-\!\s^{k\!-\!1})^{\!\top}\|_F) \\
    &\leq \sqrt{|\calU|}\left((1 \!+\! \alpha_0)C_V \!+\! \bar{\lambda}_U\right)\|\bfU^{k}\!-\!\bfU^{k\!-\!1}\|_F + (1/\gamma)\|\bfU^{k}\!-\!\bfU^{k\!-\!1}\|_F \\
    &\phantom{\leq} + \rho|\calU|^{-1}\| \1(\w^{k}\!-\!\w^{k\!-\!1})^{\!\top}\|_F+\! \rho |\calU|^{-1}\|\1(\s^{k}\!-\!\s^{k\!-\!1})^{\!\top}\|_F,
\end{align*}
where the second equality follows from the fact that $\bfU^{k}$ minimises $\nicefrac{\rho}{2}\||\calU|^{-1}\bfU^\top\1 \!+\! \w^{k\!-\!1} \!-\! \s^{k\!-\!1}\|_2^2 \!-\! \nicefrac{\rho}{2} \|\w^{k\!-\!1}\|_2^2 + \nicefrac{1}{2\gamma}\|\bfU\!-\!\bfU^{k\!-\!1}\|_F^2 \!+\! \langle \bfU-\bfU^{k\!-\!1}, \nabla_{\bfU} g(\bfV^{k}\!, \bfU^{k\!-\!1}\!, \s^{k\!-\!1})\rangle_F$.
Here, since $\|\bfU^{k}-\bfU^{k\!-\!1}\|_F$, $\|\w^{k}-\w^{k\!-\!1}\|_2$ converge to $0$, $\nabla_{\bfU}L_{\rho}(\bfV^{k},\bfU^{k},\s^{k},\w^{k})$ converges to $0$.

Because $\s^{k}$ minimises $L_{\rho}(\bfV^{k},\bfU^{k},\s,\w^{k\!-\!1})$, we also have 
\begin{align*}
    \!\|\nabla_{\s}L_{\rho}(\bfV^{k}\!,\bfU^{k}\!,\s^{k}\!,\w^{k})\|_2 \!&=\! \|\nabla_{\s}L_{\rho}(\!\bfV^{k}\!,\bfU^{k}\!,\s^{k}\!,\w^{k\!-\!1}\!) \!-\! \nabla_{\s}L_{\rho}(\!\bfV^{k}\!,\bfU^{k}\!,\s^{k}\!,\w^{k}\!)\|_2 \\
    &= \rho\|\w^{k} - \w^{k\!-\!1}\|_2.
\end{align*}
Thus, since $\|\w^{k} - \w^{k-1}\|_2$ converges to $0$, $\nabla_{\s}L_{\rho}(\bfV^{k},\bfU^{k},\s^{k},\w^{k})$ converges to $0$.
Finally, it holds that
\begin{align*}
    \|\nabla_{\w}L_{\rho}(\bfV^{k},\bfU^{k},\s^{k},\w^{k})\|_2 &\!=\! \rho\||\calU|^{-1}(\bfU^{k})^\top \1 \!-\! \s^{k}\|_2 \\
    &\!=\! \rho\|\w^{k} \!-\! \w^{k-1}\|_2,
\end{align*}
and hence $\nabla_{\w}L_{\rho}(\bfV^{k},\bfU^{k},\s^{k},\w^{k})$ converges to $0$.
\end{proof}

\cref{thm:fiADMM} illustrates that the sequence $\{\bfU^k, \bfV^k, \s^k, \w^k\}$ will converge to the feasible set, in which $\s=(1/|\calU|)\bfU^\top\1$ holds.
Moreover, the derivative of the augmented Lagrangian with respect to the primal variables will converge to zero, which implies that the limit points of $\{\bfU^k, \bfV^k, \s^k, \w^k\}$ should be saddle points (i.e. KKT points of \cref{eq:fiadmm-obj}) of $L_\rho$ if there exist.
Notably, the above convergence relies strongly on the fact that the objective is strongly convex with respect to each variable when the other variables are held constant.
This property is inherited from iALS, and therefore fiADMM takes advantages of iALS in both scalability and convergence.

The proofs of the lemmas are provided in the Appendix.

\section{Numerical Experiment}
\subsection{Setup}
We evaluate our fiADMM following the protocol of \citet{rendle2021revisiting} on the MovieLens 20M (\ML)~\cite{harper2015movielens} and Million Song Dataset (\MSD)~\cite{bertin2011million} benchmarks.
The evaluation procedure follows a strong generalisation setting, in which we use all interactions of 80\% of the users for training and consider the remaining two sets of 10\% of the users as holdout splits.
In the validation and testing phases, a system predicts the preference scores of all items for each user based on the 80\% interactions of the user to produce the ranked list and then computes ranking measures using the remaining 20\% of the interactions for the ranked list.
As systems must make predictions for users who do not appear in the training phase,
fiADMM optimises $\bfU,\s,\w$ with the fixed $\bfV$ based on the users' 80\% of the users' interactions as in iALS.
Throughout the experiments, we trained fiADMM by setting $T\!=\!100$ and $T\!=\!50$ for the number of training and prediction epochs, respectively, with a constant learning rate $\gamma\!=\!0.05$ and standard deviation $\sigma\!=\!0.1$ for initialisation.
For fiADMM, we tuned $\lambda_2$, $\alpha_0$, $\lambda_f$, and $\rho$ as hyper-parameters; we set exponent $\eta\!=\!1.0$ in L2 regularisation for all settings of iALS and fiADMM.
We implemented fiADMM\footnote{We will publish the code here.} based on the efficient C++ implementation provided by \citet{rendle2021revisiting}\footnote{\url{https://github.com/google-research/google-research/tree/master/ials/}}, which is multi-threaded and uses Eigen\footnote{\url{https://eigen.tuxfamily.org/}} for vector and matrix operations that support AVX instructions.
For a fair comparison, we used frequency-based re-scaling of $\bLambda_U$ and $\bLambda_V$~\cite{zhou2008large,rendle2021revisiting} for both iALS and fiADMM.

\subsection{Final Quality}
\noindent\textbf{Trade-off between Quality and Fairness.}
As fiADMM is designed to balance ranking quality and item fairness, we first examine its trade-off efficiency comparing with iALS.
We evaluate models obtained in a grid-search of hyper-parameters; $\lambda_2$ and $\alpha_0$ for iALS, and $\lambda_2$, $\alpha_0$, $\lambda_f$, and $\rho$ for fiADMM\footnote{Hyper-parameters are from $\lambda_2 \!\in\!\! \{1\mathrm{e-}3,\! 2\mathrm{e-}3,\! \dots,\! 9\mathrm{e-}3 \}$, $\alpha_0 \!\in\! \{ 0.01,\! 0.02,\! \dots,\! 0.09,\! 0.1,\! 0.2 \}$, and $\lambda_f,\rho \!\in\! \{1\mathrm{e}3,\! 2\mathrm{e}3,\! \dots,\! 9\mathrm{e}3,\! 1\mathrm{e}4,\! 1.3\mathrm{e}4,\! 1.5\mathrm{e}4,\! 1.7\mathrm{e}4,\! 2\mathrm{e}4 \!\}$.}.

\cref{fig:trade_off_quality_vs_fairness} summarises the trade-off between ranking quality and item fairness on the validation splits of \ML and \MSD.
We use R@$K$ ($K\!=\!20,50$) and nDCG@100 as measures of ranking quality and Gini@$K$ as those of the inequality of item exposure.
The Gini@$K$ measure is the Gini index by defining the utility of item $j$ as $o_j \!=\! \sum_{i \in \calU}\ind{\text{$j$ in the top-$K$ for $i$}}$ in \cref{eq:gini-index}.
The lines in each figure indicate the Pareto frontiers of methods with various hyper-parameter settings.
In \ML, fiADMM (red line) clearly achieve a more satisfactory trade-off than iALS, particularly in settings with strict fairness constraints (left side of the figures).
However, fiADMM compromises ranking quality more than iALS when item fairness is not important (right side in the figures).
The results of \MSD are more clear than those of \ML: fiADMM demonstrates a superior efficiency to iALS with a large margin.
The degradation in the quality-heavy settings may be because the inequality of item popularity is larger in \ML than in \MSD; the Gini indices of the number of training interactions for items are $0.901$ and $0.558$ for \ML and \MSD, respectively.
Thus, quality and fairness are not in severe conflict with each other in \MSD, whereas achieving high fairness is difficult on \ML under a strict constraint on quality.

\begin{figure}[t]
    \centering
    \includegraphics[keepaspectratio, width=0.99\linewidth]{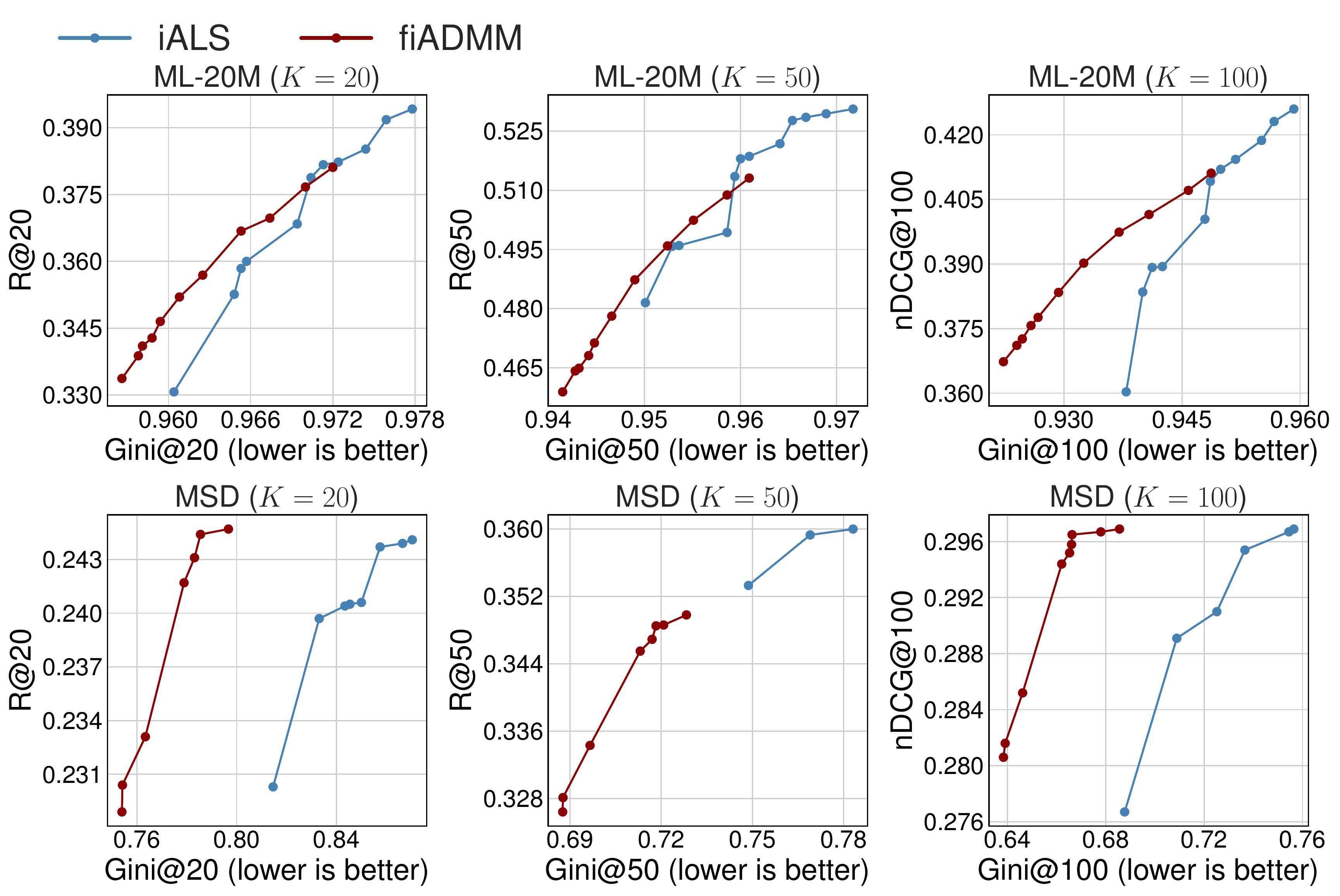}
    \caption{Trade-off between quality and fairness.}
    \vspace{-0.2cm}
    \label{fig:trade_off_quality_vs_fairness}
\end{figure}

\noindent\textbf{Distribution of Exposure.}
To visualise the extent to which fiADMM improves item fairness,
\cref{fig:exposure_vs_coverage} illustrates the distribution of item exposure in iALS and fiADMM on \ML (top row) and \MSD (bottom row).
Each curve in the figures shows the exposure of items in the top-$K$ rankings sorted in increasing order.
We omitted the items that did not appear in any top-$K$ rankings and hence the right end points on the x-axis represent the numbers of unique items that are exposed at least once, that is, item coverage.
The curves in the figures represent the iALS with the best setting in terms of ranking quality and fiADMM with small and large $\lambda_f$, $\lambda_f\!=\!1\mathrm{e}3$ and $\lambda_f\!=\!1\mathrm{e}4$ for \ML and $\lambda_f\!=\!5\mathrm{e}3$ and $\lambda_f\!=\!5\mathrm{e}4$ for \MSD.
The overall trend is clear: for various $K$, fiADMM improves item coverage and reduces the maximum exposure for an item when using a large $\lambda_f$.
Moreover, the fiADMM models with the large $\lambda_f$ also retain acceptable nDCG@100; $0.378$ on \ML and $0.283$ on \MSD while those of the iALS models are $0.426$ on \ML and $0.297$ on \MSD.
This suggests the effectiveness of our score-based fairness regulariser as a surrogate of the exposure inequality in rankings.

\noindent\textbf{Combining Post-Processing.}
Because fiADMM is an in-processing method, it can be utilised with some post-processing methods.
We hence investigate the effect of applying fairness-aware post-processing to fiADMM.
This experiment considers a method that combines fiADMM and FairRec~\cite{patro2020fairrec} (fiADMM+FairRec), in which 
we first train an MF model using fiADMM for preference estimation and then optimise allocation of ranked items based on the FairRec algorithm\footnote{We jointly tuned $\alpha_0, \lambda_2, \lambda_f, \rho$, and additionally the scale $l \in (0, 1]$ of the minimum allocation constraint for each item $l \cdot (K \cdot |\calU|) / |\calV|$ where $K\!=\!100$. We tuned $l$ in the range of $\{1\mathrm{e-}3,2\mathrm{e-}3,\dots,9\mathrm{e-}3, 1\mathrm{e-}2,2\mathrm{e-}2,\dots,9\mathrm{e-}2, 0.1, 0.2, \dots, 0.9\}$.}.
\cref{fig:tradeoff-with-fairrec} shows the Pareto-frontiers of fiADMM and fiADMM+FairRec on \ML.
The effect of applying FairRec is not advantageous in the balance of nDCG@$100$ and Gini@$100$,
whereas it substantially improves that of R@$20$ and Gini@$20$ in fairness-heavy settings;
in the quality-heavy settings, the performance gain is not considerable even in the left figure.
Because FairRec is not parallelisable with respect to users and thus computationally taxing for large-scale settings,
fiADMM is a reasonable choice to balance quality and fairness in practise.

\begin{figure}[t]
    \centering
    \includegraphics[keepaspectratio, width=0.99\linewidth]{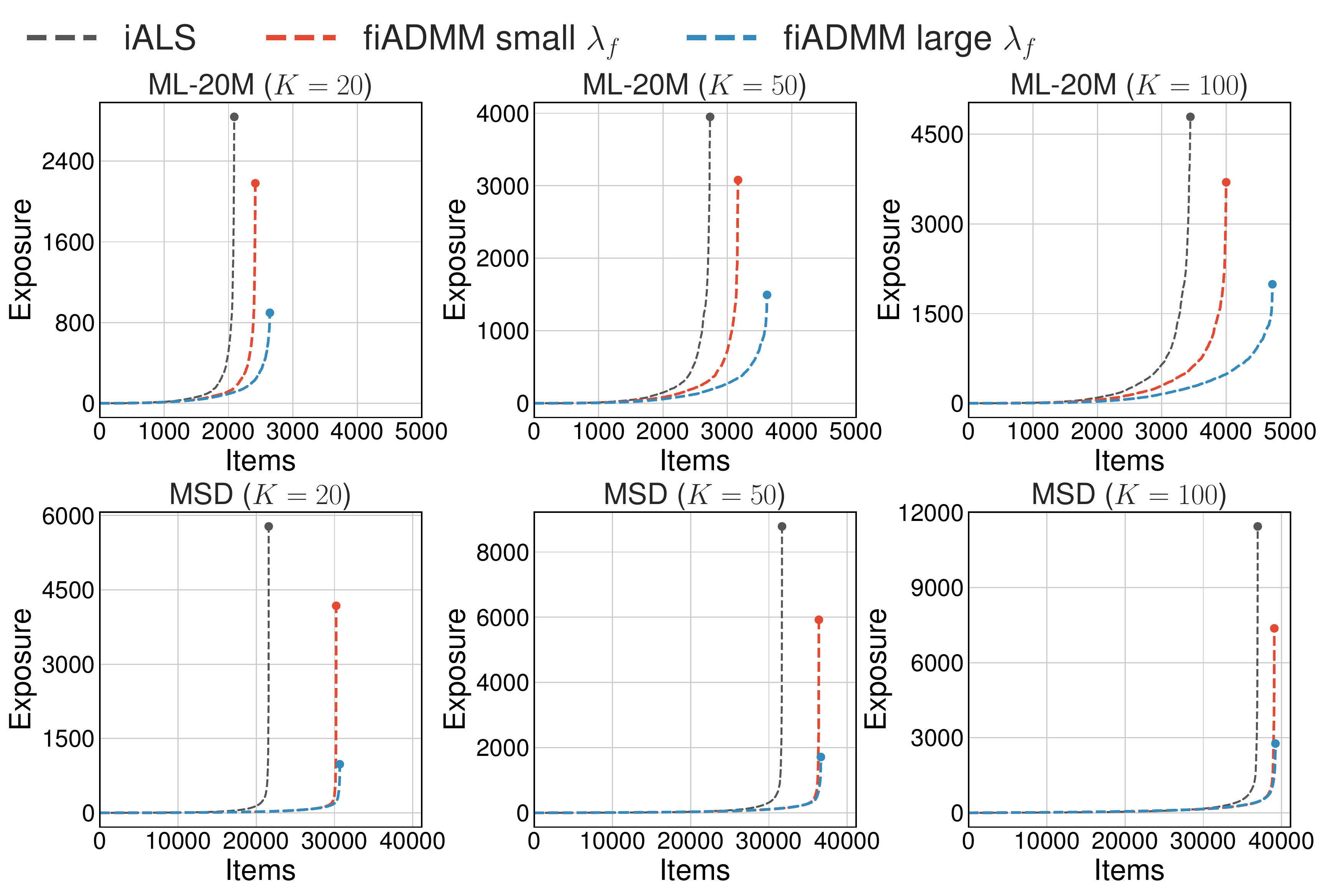}
    \caption{Distribution of item exposure.}
     \vspace{-0.2cm}
    \label{fig:exposure_vs_coverage}
\end{figure}

\begin{figure}[t]
    \centering
    \includegraphics[keepaspectratio, width=0.8\linewidth]{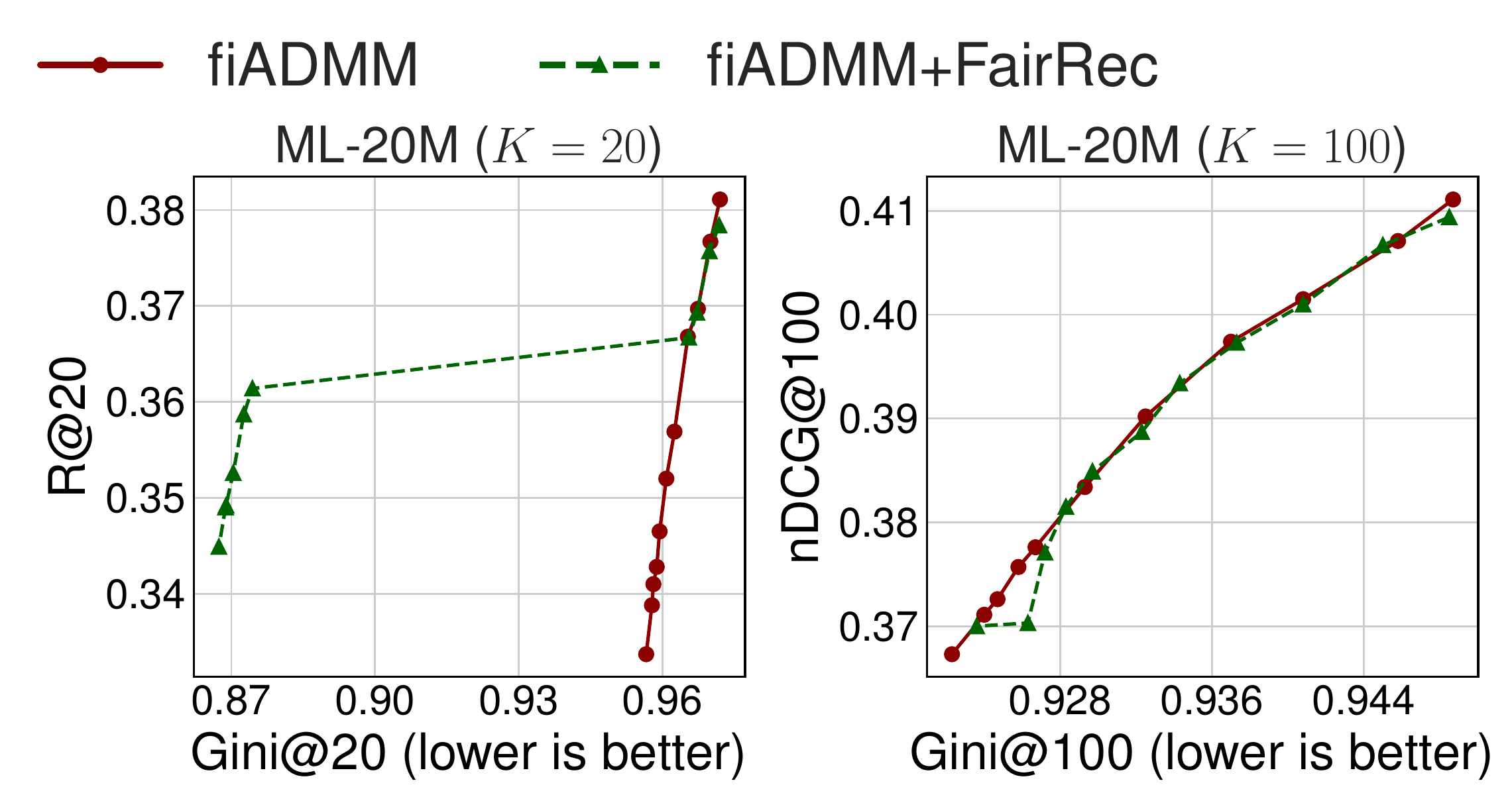}
    \caption{Effect of combining pre- and post-processing.}
    \label{fig:tradeoff-with-fairrec}
\end{figure}

\begin{table}[t]
\vspace{-0.2cm}
\caption{Ranking quality results on (a) \ML and (b) \MSD}
\vspace{-0.2cm}
\footnotesize	
\begin{subtable}[t]{0.9\linewidth}
  \caption{\ML}
  \vspace{-0.1cm}
  \label{table:effectiveness_ml20m}
  \setlength\tabcolsep{3.0pt}
   \begin{tabular}{|l|ccc|c|}
      \hline
      Method & Recall@20       & Recall@50 & nDCG@100 & Result from \\\hline %
      RecVAE~\cite{shenbin2020recvae}                    & $0.414$      & $0.553$         & $0.442$         & \cite{shenbin2020recvae}        \\
      Mult-VAE~\cite{liang2018variational}                    & $0.395$      & $0.537$         & $0.426$         & \cite{liang2018variational}        \\
      iALS                    & $0.395$      & $0.532$         & $0.425$         & \cite{rendle2021revisiting}        \\      
      $\textbf{fiADMM}$                       & $\textbf{0.377}$      & $\textbf{0.513}$         & $\textbf{0.403}$           & our result        \\
      Popularity                    & $0.162$     & $0.235$         & $0.191$         & \cite{steck2019embarrassingly}  \\\hline
  \end{tabular}
\end{subtable} \\
\begin{subtable}[t]{0.9\linewidth}
  \vspace{0.1cm}
  \caption{\MSD}
  \vspace{-0.1cm}
  \label{table:effectiveness_msd}
  \setlength\tabcolsep{3.0pt}
    \begin{tabular}{|l|ccc|c|}
      \hline
      Method & Recall@20       & Recall@50 & nDCG@100 & Result from \\\hline %
      RecVAE                  & 0.333      & 0.428         & 0.389         & \cite{shenbin2020recvae}        \\
      iALS++ ($d=8,192$)                    & 0.309      & 0.415         & 0.368        & \cite{rendle2021ials++}       \\
      Mult-VAE                    & 0.266     & 0.364         & 0.319         & \cite{liang2018variational}        \\ 
      iALS ($d=512$)                   & $0.245$      & $0.359$         & $0.297$         & our result        \\
      $\textbf{fiADMM}$ ($d=512$)                    & $\textbf{0.245}$      & $\textbf{0.349}$        & $\textbf{0.296}$         & our result        \\     
      Popularity                    & 0.043     & 0.068        & 0.058         & \cite{steck2019embarrassingly}   \\ \hline
  \end{tabular}
\end{subtable}
\end{table}

\begin{figure*}[th]
    \centering
    \includegraphics[keepaspectratio, width=0.98\linewidth]{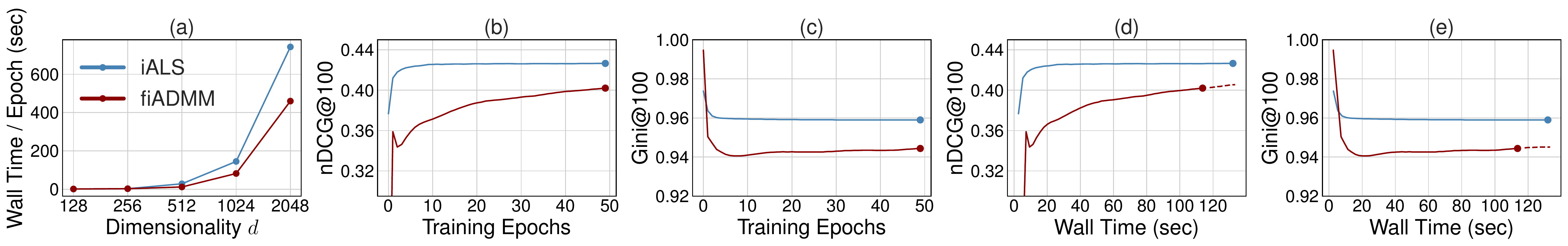}
    \caption{Training efficiency of iALS and fiADMM.}
     \vspace{-0.25cm}
    \label{fig:training_efficiency}
\end{figure*}

\begin{figure*}[th]
    \centering
    \includegraphics[keepaspectratio, width=0.98\linewidth]{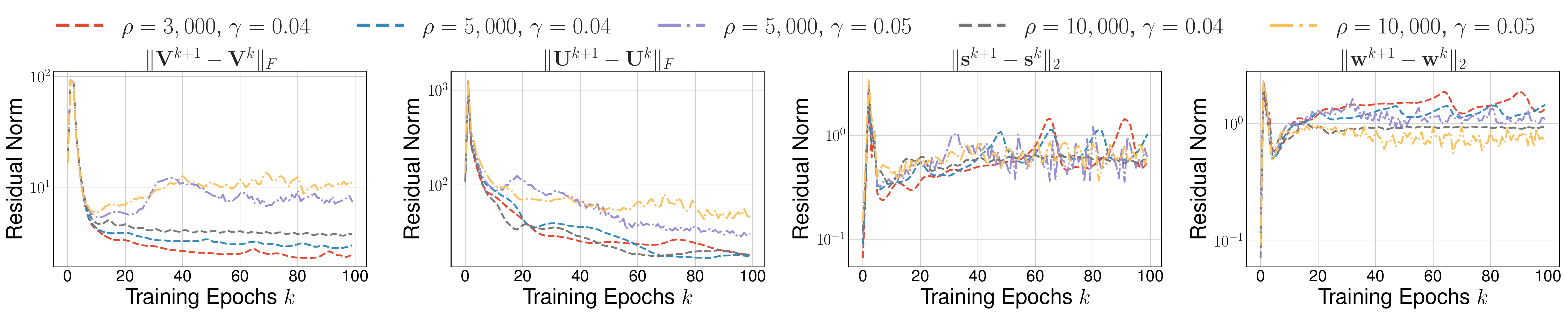}
    \caption{Convergence profile of fiADMM with varying $\rho$ and $\gamma$.}
     \vspace{-0.17cm}
    \label{fig:fiadmm_convergence}
\end{figure*}

\noindent\textbf{Ranking Effectiveness.}
We also compare fiADMM with baseline methods, namely, a pure popularity recommender~\cite{steck2019embarrassingly}, iALS~\cite{hu2008collaborative}, and variational 
autoencoders (i.e. RecVAE~\cite{shenbin2020recvae} and Mult-VAE~\cite{liang2018variational}).
For fiADMM, we set $\lambda_f\!=\!1\mathrm{e}3$ and $\lambda_f\!=\!3\mathrm{e}3$ for \ML and \MSD, respectively.
To reduce the experimental burden due to the $O(d^3)$ factor in fiADMM and iALS,
we set $d\!=\!256$ and $d\!=\!512$ for \ML and \MSD, respectively; we also report the result of an iALS variant for large $d$ (i.e. \emph{iALS++}~\cite{rendle2021ials++}) with $d\!=\!8,192$ for \MSD to show the full potential of iALS.
\cref{table:effectiveness_ml20m} and \cref{table:effectiveness_msd} show the results\footnote{All the results reported in \cref{table:effectiveness_ml20m} and \cref{table:effectiveness_msd}, where obtained by repeating the experiments ten times and the average is reported here.} for \ML and \MSD, respectively.
In \ML, the ranking quality of fiADMM is comparable to that of the baselines, whereas a substantial discrepancy is observed compared to RecVAE and iALS++ in \MSD.
The deterioration of fiADMM from iALS is not severe in both \ML and \MSD despite the fairness-aware multi-objective optimisation.
The results of iALS++ also suggest further improvement in fiADMM by using a larger $d$ in \MSD.

\subsection{Training Efficiency}

\noindent\textbf{Wall Time per Epoch.}
\label{section:wall-runtime-per-epoch}
The analysis discussed in \cref{section:complexity-analysis} suggests that fiADMM is more efficient than iALS regarding the asymptotic runtime per epoch.
\cref{fig:training_efficiency}~(a) shows the effect of $d$ on the wall time per epoch on \ML.
The x- and y-axes, respectively, indicate the dimensionality $d$ of user/item latent factors and the wall runtime per epoch in seconds\footnote{These times were obtained on a GCP instance with 57.6 GB memory and 64 vCPUs.}.
The difference between iALS and fiADMM is more pronounced for a larger $d$.
These results are clearly in agreement with the theoretical analysis.

\noindent\textbf{Convergence Speed.}
The optimisation efficiency also depends on the convergence speed in practise, that is, the number of epochs the algorithm requires to achieve acceptable performance.
\cref{fig:training_efficiency} (b)-(e) show the validation nDCG@100 and Gini@100 of iALS and fiADMM at each training epoch and training wall time on \ML for the settings of $d\!=\!256$.
In contrast to iALS, fiADMM requires many iterations to achieve high ranking quality (see (b) and (c)).
In \cref{fig:training_efficiency} (d)-(e), fiADMM also requires longer training time to achieve acceptable quality; the right end points in the figures correspond to the models after 50 iterations.
This result demonstrates the less-than-optimal convergence speed of fiADMM.

\noindent\textbf{Convergence Behaviour.}
We examine the convergence of each variable and the training losses in fiADMM with $\lambda_f\!=\!1\mathrm{e}3$ and $d\!=\!256$ on \ML.
Each figure in the top row of \cref{fig:fiadmm_convergence} demonstrates the residual norm of each variable, i.e. $\|\bfV^{k+1} \!-\! \bfV^k\|_F$, $\|\bfU^{k+1} \!-\! \bfU^k\|_F$, $\|\s^{k+1} \!-\! \s^k\|_2$, and $\|\w^{k+1} \!-\! \w^k\|_2$.
The following observations can be made: (1) fiADMM with $\gamma\!=\!0.04$ is more stable than that with $\gamma\!=\!0.05$ in terms of $\|\bfV^{k+1} \!-\! \bfV^k\|_F$ and $\|\bfU^{k+1} \!-\! \bfU^k\|_F$; and (2) small values of $\rho$ (i.e. $\rho\!=\!3\mathrm{e}3,5\mathrm{e}3$) destabilise the behaviour in $\|\s^{k+1} \!-\! \s^k\|_2$ and $\|\w^{k+1} \!-\! \w^k\|_2$.
Because the learning rate $\gamma$ directly affects $\|\bfU^{k+1} \!-\! \bfU^k\|_F$, observation (1) is as expected.
However, the residual norm of each variable is considerably small for all settings of $\rho$ and $\gamma$ regarding the large dimensionality.
Here, the update step of $\w$ is $\w^{k+1}\!=\!\w^{k}+((1/|\calU|)(\bfU^{k+1})^\top \1\!-\!\s^{k+1})$, and hence $\|\w^{k+1} \!-\! \w^k\|_2$ represents the constraint violation $\|\s^{k+1}\!-\!(1/|\calU|)(\bfU^{k+1})^\top \1\|_2$.
Therefore, observation (2) suggests that the constraint $\s\!=\!(1/|\calU|)\bfU^\top \1$ is more strictly satisfied using a larger $\rho$.

\section{Conclusion}
The feasibility of fairness-aware item recommendation is indispensable for solving fairness issues; however, it is rather disregarded in academic research.
Hence, this study aimed to develop a simple and scalable tool for solving immediate unfairness issues in real-world applications.
We proposed fiADMM, a variant of iALS with the item fairness regulariser.
Despite the difficulty of optimising fairness regularisation in parallel,
fiADMM ultimately maintains scalability and provable convergence.
These properties are realised by carefully designing (1) a tractable fairness regulariser based on predicted merit, (2) variable splitting using three-block reformulation, and (3) linearisation and efficient proximal mapping to combat optimisation non-separability among users.
In future, we plan to analyse the convergence rate and further enhance fiADMM in terms of its convergence speed.
It would also be interesting to extend fiADMM to enable the use of a large $d$.

%% file: appendix_double_column.tex

\section{Proofs for Theorem \ref{thm:fiADMM}}

\subsection{Proof of Lemma \ref{lem:vu_sub_ub}}
\begin{proof}
From the definition of $L_\rho(\bfV, \bfU, \s^{k}, \w^{k})$, we have:
\begin{align}
\label{eq:vu_sub_eq}
    &L_\rho(\bfV^{k+1}, \bfU^{k+1}, \s^{k}, \w^{k}) - L_\rho(\bfV^{k}, \bfU^{k}, \s^{k}, \w^{k}) \nonumber\\
    &= g(\bfV^{k+1}, \bfU^{k+1}, \s^{k}) + \frac{\rho}{2}\norm{\w^{k} + |\frac{1}{|\calU|}(\bfU^{k+1})^\top\1 - \s^{k}}_2^2 - \frac{\rho}{2} \|\w^{k}\|_2^2 \nonumber\\
    &\phantom{=} - g(\bfV^{k}, \bfU^{k}, \s^{k}) - \frac{\rho}{2}\norm{\w^{k} + \frac{1}{|\calU|}(\bfU^{k})^\top \1 - \s^{k}}_2^2 + \frac{\rho}{2} \|\w^{k}\|_2^2 \nonumber\\
    &= g(\bfV^{k+1}, \bfU^{k+1}, \s^{k}) \!-\! g(\bfV^{k+1}, \bfU^{k}, \s^{k}) \!+\! g(\bfV^{k+1}, \bfU^{k}, \s^{k}) \!-\! g(\bfV^{k}, \bfU^{k}, \s^{k}) \nonumber\\
    &\phantom{=} + \frac{\rho}{2}\norm{\w^{k} \!+\! \frac{1}{|\calU|}(\bfU^{k+1})^\top \1 \!-\! \s^{k}}_2^2 \!-\! \frac{\rho}{2}\norm{\w^{k} \!+\! \frac{1}{|\calU|}(\bfU^{k})^\top\1 \!-\! \s^{k}}_2^2 .
\end{align}
Denoting the Gramian matrix of $\bfU$ by $\bfG_U=\bfU^\top \bfU$,
we have
\begin{align*}
    &\langle \nabla_{\bfV} g(\bfV, \bfU, \s) - \nabla_{\bfV} g(\bfV', \bfU, \s), \bfV - \bfV'\rangle_F \\
    &= \sum_{j \in \calV}\!\Biggl\langle\!\left(\sum_{i \in \calU}r_{i,j}\bu_i\bu_i^{\top} + \alpha_0\bfG_U + \lambda_f\s\s^{\top} + \lambda_V^{(j)}\I\right)\bv_j - \sum_{i \in \calU}r_{i,j}\bu_i \\
    &\phantom{=} - \left(\sum_{i \in \calU}r_{i,j}\bu_i\bu_i^{\top} + \alpha_0\bfG_U + \lambda_f\s\s^{\top} + \lambda_V^{(j)}\I\right)\bv_j' + \sum_{i \in \calU}r_{i,j}\bu_i, \bv_j \!-\! \bv_j' \Biggr\rangle \\
    &= \sum_{j \in \calV}\!\!\left\langle\!\left(\sum_{i \in \calU}r_{i,j}\bu_i\bu_i^{\top}\right)\!(\bv_j\!-\!\bv_j'), \bv_j \!-\! \bv_j' \right\rangle \!+\! \sum_{j \in \calV}\!\!\left\langle \alpha_0\bfG_U (\bv_j\!-\!\bv_j'), \bv_j \!-\! \bv_j' \right\rangle \\
    &\phantom{=} + \sum_{j \in \calV}\left\langle\lambda_f\s\s^{\top}(\bv_j\!-\!\bv_j'), \bv_j \!-\! \bv_j' \right\rangle + \sum_{j \in \calV}\left\langle \lambda_V^{(j)}(\bv_j-\bv_j'), \bv_j \!-\! \bv_j' \right\rangle\\
    &= \sum_{j \in \calV}\sum_{i \in \calU}r_{i,j}\|\bu_i^{\top}(\bv_j\!-\!\bv_j')\|_2^2 + \alpha_0\sum_{j \in \calV}\sum_{i \in \calU}\|\bu_i^{\top}(\bv_j\!-\!\bv_j')\|_2^2 \\
    &\phantom{=} + \lambda_f\sum_{j \in \calV}(\s^{\top}(\bv_j\!-\!\bv_j'))^2 + \sum_{j \in \calV}\| \lambda_V^{(j)}(\bv_j\!-\!\bv_j')\|_2^2\\
    &\geq \underline{\lambda}_V\sum_{j \in \calV}\| \bv_j-\bv_j'\|_2^2 = \underline{\lambda}_V\| \bfV-\bfV'\|_F^2,
\end{align*}
where $\underline{\lambda}_V=\min_{j \in \calV} \lambda_V^{(j)}$.
Thus, the function $g$ is a $\underline{\lambda}_V$-strongly convex function with respect to $\bfV$.
We also have
\begin{align}
    \label{eq:v_sub_ub}
    &g(\bfV^{k+1}\!,\bfU^{k}\!,\s^{k}) \!-\! g(\bfV^{k}\!,\bfU^{k}\!,\s^{k}) \nonumber \\ 
    &\leq \langle \nabla_{\bfV}g(\bfV^{k+1}\!,\bfU^{k}\!,\s^{k})\!, \bfV^{k+1} \!-\! \bfV^{k}\rangle
    \!-\! \frac{\underline{\lambda}_V}{2}\|\bfV^{k+1} \!-\! \bfV^{k}\|_F^2 \nonumber\\
    &=  - \frac{\underline{\lambda}_V}{2}\|\bfV^{k+1} \!-\! \bfV^{k}\|_F^2,
\end{align}
where the last equality follows from the fact that $\bfV^{k+1}$ minimises $g(\bfV, \bfU^{k}, \s^{k})$; hence $\nabla_{\bfV}g(\bfV^{k+1},\bfU^{k}, \s^{k})=0$ holds.
Moreover, since $\bfU^{k+1}$ minimises $(\rho/2)\|\w^{k} + |\calU|^{-1}(\bfU)^\top \1 - \s^{k}\|_2^2 - (\rho/2) \|\w^{k}\|_2^2 + (1/2\gamma)\|\bfU-\bfU^{k}\|_F^2 + \langle \bfU-\bfU^{k}, \nabla_{\bfU} g(\bfV^{k+1}, \bfU^{k}, \s^{k})\rangle_F$:
\begin{align}
    &\frac{\rho}{2}\|\w^{k} + \frac{1}{|\calU|}(\bfU^{k+1})^\top \1 - \s^{k}\|_2^2 - \frac{\rho}{2} \|\w^{k}\|_2^2 + \frac{1}{2\gamma}\|\bfU^{k+1}-\bfU^{k}\|_F^2 \nonumber \\
    &\phantom{\leq} + \langle \bfU^{k+1}-\bfU^{k}, \nabla_{\bfU} g(\bfV^{k+1}, \bfU^{k}, \s^{k})\rangle \nonumber\\
    \label{eq:u_sub_ub}
    &\leq \frac{\rho}{2}\|\w^{k} + \frac{1}{|\calU|}(\bfU^{k})^\top \1 - \s^{k}\|_2^2 - \frac{\rho}{2} \|\w^{k}\|_2^2.
\end{align}
By combining \cref{eq:vu_sub_eq,eq:v_sub_ub,eq:u_sub_ub}, we get:
\begin{align}
\label{eq:vu_sub}
    &L_\rho(\bfV^{k+1}, \bfU^{k+1}, \s^{k}, \w^{k}) - L_\rho(\bfV^{k}, \bfU^{k}, \s^{k}, \w^{k})  \nonumber\\
    &\leq g(\bfV^{k+1}, \bfU^{k+1}, \s^{k}) \!-\! g(\bfV^{k+1}, \bfU^{k}, \s^{k}) \!-\! \langle \bfU^{k+1}\!-\!\bfU^{k}, \nabla_{\bfU} g(\bfV^{k+1}, \bfU^{k}, \s^{k})\rangle_F\nonumber\\
    &\phantom{\leq} - \frac{1}{2\gamma}\|\bfU^{k+1}-\bfU^{k}\|_F^2 - \frac{\underline{\lambda}_V}{2}\|\bfV^{k+1}-\bfV^{k}\|_F^2.
\end{align}
On the other hand, under the assumption in Theorem~\ref{thm:fiADMM}, from Lemma~\ref{lem:g_smooth}, the function $g(\bfV^{k+1},\bfU, \s^{k})$ is a $\sqrt{|\calU|}\left((1 + \alpha_0)C_V + \bar{\lambda}_U\right)$-smooth function with respect to $\bfU$.
Then, we have for any $\bfU,\bfU'$:
\begin{align}
      \nonumber
     &g(\bfV^{k+1}, \bfU', \s^{k}) - g(\bfV^{k+1}, \bfU, \s^{k}) 
     - \langle \nabla_{\s}g(\bfV^{k+1}, \bfU, \s^{k}), \bfU'- \bfU\rangle_F \\
     \label{eq:g_smooth_u}
     &\leq \frac{\sqrt{|\calU|}((1 + \alpha_0)C_V + \bar{\lambda}_U)}{2}\|\bfU-\bfU'\|_F^2.
\end{align}
By combining \cref{eq:vu_sub} and \cref{eq:g_smooth_u}, we get:
\begin{align*}
    &L_\rho(\bfV^{k+1}, \bfU^{k+1}, \s^{k}, \w^{k}) - L_\rho(\bfV^{k}, \bfU^{k}, \s^{k}, \w^{k}) \\
    &\leq \!\frac{\sqrt{|\calU|}((1 \!+\! \alpha_0)C_V \!+\! \bar{\lambda}_U) \!-\! 1/\gamma}{2}\|\bfU^{k+1}\!-\!\bfU^{k}\|_F^2- \frac{\underline{\lambda}_V}{2}\|\bfV^{k+1}-\bfV^{k}\|_F^2.
\end{align*}
\end{proof}

\subsection{Proof of Lemma \ref{lem:s_sub_ub}}
\begin{proof}
Let us define $h^{k}(\s)=g(\bfV^{k+1},\bfU^{k+1},\s) + \frac{\rho}{2}\|\w^{k} + |\calU|^{-1}(\bfU^{k+1})^\top \1 - \s\|_2^2 - \frac{\rho}{2} \|\w^{k}\|_2^2$.
We have:
\begin{align*}
    &\langle \nabla h^{k}(\s) - \nabla h^{k}(\s'), \s - \s'\rangle \\
    &= \langle \nabla_{\s} g(\bfV^{k+1},\bfU^{k+1},\s) -\rho(\w^{k}+\frac{1}{|\calU|}(\bfU^{k+1})^\top \1-\s) \\
    &\phantom{=} - \nabla_{\s} g(\bfV^{k+1},\bfU^{k+1},\s') + \rho(\w^{k} + \frac{1}{|\calU|}(\bfU^{k+1})^\top \1 - \s'), \s - \s'\rangle \\
    &= \langle \nabla_{\s} g(\bfV^{k+1},\bfU^{k+1},\s)  - \nabla_{\s} g(\bfV^{k+1},\bfU^{k+1},\s') + \rho(\s - \s'), \s - \s'\rangle \\
    &\geq \rho\|\s-\s'\|_2^2,
\end{align*}
where the inequality follows from the convexity of $g(\bfV^{k+1},\bfU^{k+1},\s)$.
Thus, $h^{k}$ is a $\rho$-strongly convex function.
Therefore, we have:
\begin{align*}
    &L_\rho(\bfV^{k+1}\!, \bfU^{k+1}\!, \s^{k+1}\!, \w^{k}) - L_\rho(\bfV^{k+1}\!, \bfU^{k+1}\!, \s^{k}\!, \w^{k}) \\
    &= h^{k}(\s^{k+1}) - h^{k}(\s^{k}) \\
    &\leq \langle \nabla h^{k}(\s^{k+1})\!, \s^{k+1}-\s^{k}\rangle - \frac{\rho}{2}\|\s^{k+1}-\s^{k}\|_2^2 = - \frac{\rho}{2}\|\s^{k+1}-\s^{k}\|_2^2,
\end{align*}
where the last equality follows from that $\s^{k+1}$ minimises $h^{k}(\s)$, i.e. $\nabla h^{k}(\s^{k+1})=0$.
\end{proof}

\subsection{Proof of Lemma \ref{lem:w_sub_ub}}
\begin{proof}
From the definition of $L_\rho(\bfV^{k+1}, \bfU^{k+1}, \s^{k+1}, \w)$ and the update rule of $\w^{k}$, we have:
\begin{align}
\label{eq:w_sub_eq}
    &L_\rho(\bfV^{k+1}, \bfU^{k+1}, \s^{k+1}, \w^{k+1}) - L_\rho(\bfV^{k+1}, \bfU^{k+1}, \s^{k+1}, \w^{k}) \nonumber\\
    &= \frac{\rho}{2}\|\w^{k+1} + \frac{1}{|\calU|}(\bfU^{k+1})^\top\1 - \s^{k+1}\|_2^2 - \frac{\rho}{2} \|\w^{k+1}\|_2^2 \nonumber \\
    &\phantom{=} - \frac{\rho}{2}\|\w^{k} + \frac{1}{|\calU|}(\bfU^{k+1})^\top\1 - \s^{k+1}\|_2^2 + \frac{\rho}{2} \|\w^{k}\|_2^2 \nonumber\\
    &= \rho\langle \w^{k+1} - \w^{k}, \frac{1}{|\calU|}(\bfU^{k+1})^\top\1-\s^{k+1}\rangle = \rho\|\w^{k+1}-\w^{k}\|_2^2.
\end{align}
On the other hand, since $\s^{k+1}$ minimises the convex function $h^{k}(\s)$, the first-order optimality condition gives:
\begin{align*}
    \nabla h^{k}(\s^{k+1}) &= \nabla_{\s}g(\bfV^{k+1}, \bfU^{k+1}, \s^{k+1}) - \rho(\w^{k}+\frac{1}{|\calU|}(\bfU^{k+1})^\top\1-\s^{k+1}) \\
    &= \nabla_{\s}g(\bfV^{k+1}, \bfU^{k+1}, \s^{k+1}) - \rho \w^{k+1}= 0.
\end{align*}
Thus,
\begin{align}
\label{eq:w_eq}
    \w^{k+1} = \frac{1}{\rho}\nabla_{\s}g(\bfV^{k+1}, \bfU^{k+1}, \s^{k+1}).
\end{align}
By combining \cref{eq:w_sub_eq}, \cref{eq:w_eq}, and Lemma~\ref{lem:g_smooth}, we have:
\begin{align*}
    &L_\rho(\bfV^{k+1}, \bfU^{k+1}, \s^{k+1}, \w^{k+1}) - L_\rho(\bfV^{k+1}, \bfU^{k+1}, \s^{k+1}, \w^{k}) \\
    &= \frac{1}{\rho}\|\nabla_{\s}g(\bfV^{k+1}, \bfU^{k+1}, \s^{k+1}) - \nabla_{\s}g(\bfV^{k}, \bfU^{k}, \s^{k})\|_2^2 \\
    &= \frac{1}{\rho}\|\nabla_{\s}g(\bfV^{k+1}, \bfU^{k+1}, \s^{k+1}) - \nabla_{\s}g(\bfV^{k+1}, \bfU^{k}, \s^{k}) \\
    &\phantom{=} + \nabla_{\s}g(\bfV^{k+1}, \bfU^{k}, \s^{k}) - \nabla_{\s}g(\bfV^{k}, \bfU^{k}, \s^{k})\|_2^2 \\
    &\leq \frac{1}{\rho}\Bigl(\|\nabla_{\s}g(\bfV^{k+1}, \bfU^{k+1}, \s^{k+1}) \!-\! \nabla_{\s}g(\bfV^{k+1}, \bfU^{k}, \s^{k})\|_2 \\
    &\phantom{\leq} + \|\nabla_{\s}g(\bfV^{k+1}, \bfU^{k}, \s^{k}) \!-\! \nabla_{\s}g(\bfV^{k}, \bfU^{k}, \s^{k})\|_2\Bigr)^2 \\
    \!\!&\leq\!\! \frac{1}{\rho}\!\!\left(\lambda_f\|\bfV^{k+1}\|_F^2\|\s^{k+1}\!-\!\s^{k}\|_2 \!+\! \lambda_f\!\left(\!\|\bfV^{k+1}\|_F \!+\! \|\bfV^{k}\|_F\right)\!\|\bfV^{k+1}\!-\!\bfV^{k}\|_F\|\s^{k}\|_2\!\right)^2 \\
    \!\!&\leq\!\! \frac{3}{\rho}\!\!\left(\lambda_f^2\|\bfV^{k+1}\|_F^4\|\s^{k+1}\!-\!\s^{k}\|_2^2 \!+\! \lambda_f^2\!\left(\!\|\bfV^{k+1}\|_F^2 \!+\! \|\bfV^{k}\|_F^2\right)\!\|\bfV^{k+1}\!-\!\bfV^{k}\|_F^2\|\s^{k}\|_2^2\!\right) \\
    \!\!&\leq\! \frac{3\lambda_f^2}{\rho}\left(C_V^2\|\s^{k+1}\!-\!\s^{k}\|_2^2 + 2C_VC_s\|\bfV^{k+1}-\bfV^{k}\|_F^2\right),
\end{align*}
where the third inequality follows from $(a+b+c)^2\leq 3(a^2 + b^2 + c^2)$ for $a,b,c\in \mathbb{R}$.
\end{proof}

\subsection{Proof of Lemma \ref{lem:L_rho_lb}}
\begin{proof}[Proof of Lemma \ref{lem:L_rho_lb}]
Under the assumption in Theorem~\ref{thm:fiADMM}, from Lemma~\ref{lem:g_smooth}, the function $g(\bfV^{k},\bfU^{k}, \s)$ is a $\lambda_f C_V$-smooth function with respect to $\s$, and then we have for any $\s,\s'$:
\begin{align}
\label{eq:g_smooth_s}
     g(\bfV^{k}\!, \bfU^{k}\!, \s') \!-\! g(\bfV^{k}\!, \bfU^{k}\!, \s) \!-\! \langle \nabla_{\s}g(\bfV^{k}\!, \bfU^{k}\!, \s), \s' \!-\! \s\rangle \leq \frac{\lambda_f C_V}{2}\|\s \!-\! \s'\|_2^2.
\end{align}
By combining \cref{eq:w_eq} and \cref{eq:g_smooth_s}, we have:
\begin{align*}
    &L_\rho(\bfV^{k}, \bfU^{k}, \s^{k}, \w^{k}) \\
&= g(\bfV^{k}\!, \bfU^{k}\!, \s^{k}) \!+\! \frac{\rho}{2}\|\w^{k} + |\calU|^{-1}(\bfU^{k})^\top \1 \!-\! \s^{k}\|_2^2 \!-\! \frac{\rho}{2} \|\w^{k}\|_2^2 \\
&= g(\bfV^{k}\!, \bfU^{k}\!, \s^{k}) \!+\! \rho\langle \w^{k}, |\calU|^{-1}(\bfU^{k})^\top\1\!-\!\s^{k}\rangle \!+\! \frac{\rho}{2}\||\calU|^{-1}(\bfU^{k})^{\top}\1 - \s^{k}\|_2^2\\
&= g(\bfV^{k}\!, \bfU^{k}\!, \s^{k}) \!-\! \langle \nabla_{\s}g(\bfV^{k}\!, \bfU^{k}\!, \s^{k})\!, \s^{k} \!-\! |\calU|^{-1}(\bfU^{k})^{\top}\1\rangle \\
&\phantom{=} \!+\! \frac{\rho}{2}\||\calU|^{-1}(\bfU^{k})^{\top}\1 \!-\! \s^{k}\|_2^2 \\
&\geq g(\bfV^{k}, \bfU^{k}, |\calU|^{-1}(\bfU^{k})^{\top}\1) \!+\! \frac{\rho-\lambda_f C_V}{2}\||\calU|^{-1}(\bfU^{k})^{\top}\1 \!-\! \s^{k}\|_2^2 \\
&\geq \frac{\rho-\lambda_f C_V}{2}\||\calU|^{-1}(\bfU^{k})^{\top}\1-\s^{k}\|_2^2,
\end{align*}
where the last inequality follows from $g(\bfV,\bfU,\s)\geq 0$ for any $\bv$, $\bfU$, and $\s$.
\end{proof}

\subsection{Proof for Lemma \ref{lem:g_smooth}}
\begin{proof}
For fixed $\bfU,\s$, for all $\bfV,\bfV'$, we have the following
\begin{align*}
    &\|\nabla_{\bfV}g(\bfV,\bfU,\s) - \nabla_{\bfV}g(\bfV',\bfU,\s)\|_2 \\
    &=\sum_{j \in \calV}\left\|\left(\sum_{i \in \calU}r_{i,j}\bu_i\bu_i^{\top} + \alpha_0\bfG_U + \lambda_f\s\s^{\top} + \lambda_V^{(j)}\I\right)(\bv_j-\bv_j')\right\|_2 \\
    &\leq \sum_{j \in \calV}\sum_{i \in \calU}\|r_{i,j}\bu_i\bu_i^{\top}(\bv_j-\bv_j')\|_2 + \sum_{j \in \calV}\sum_{i \in \calU}\| \alpha_0\bu_i\bu_i^{\top}(\bv_j-\bv_j')\|_2 \\
    &\phantom{\leq} +  \sum_{j \in \calV}\| \lambda_f\s\s^{\top}(\bv_j-\bv_j')\| + \sum_{j \in \calV}\| \lambda_V^{(j)}(\bv_j-\bv_j')\|_2 \\
    &= \sum_{j \in \calV}\sum_{i \in \calU}r_{i,j}|\bu_i^{\top}(\bv_j-\bv_j')| \cdot \|\bu_i\|_2 + \alpha_0\sum_{j \in \calV}\sum_{i \in \calU}|\bu_i^{\top}(\bv_j-\bv_j')| \cdot \| \bu_i\|_2 \\
    &\phantom{=} + \lambda_f\sum_{j \in \calV}|\s^{\top}(\bv_j-\bv_j')| \cdot \|\s\|_2 + \sum_{j \in \calV}\|\lambda_V^{(j)}(\bv_j-\bv_j')\|_2 \\
    &\leq \sum_{j \in \calV}\sum_{i \in \calU}r_{i,j}\|\bu_i\|_2^2\|\bv_j-\bv_j'\|_2 + \alpha_0\sum_{j \in \calV}\sum_{i \in \calU}\|\bu_i\|_2^2\|\bv_j-\bv_j'\|_2 \\
    &\phantom{\leq} + \lambda_f\sum_{j \in \calV}\|\s\|_2^2\|\bv_j-\bv_j'\|_2 + \bar{\lambda}_V \sum_{j \in \calV}\|(\bv_j-\bv_j')\|_2 \\
    &\leq \left(\left(1 + \alpha_0\right)\sum_{i \in \calU}\|\bu_i\|_2^2 + \|\s\|_2^2 + \bar{\lambda}_V\right)\sum_{j \in \calV}\|\bv_j-\bv_j'\|_2 \\
    &\leq \sqrt{|\calV|}\left(\left(1 + \alpha_0\right)\|\bfU\|_F^2 + \|\s\|_2^2 + \bar{\lambda}_V\right)\|\bfV-\bfV'\|_2,
\end{align*}
where $\bar{\lambda}_V=\max_{j \in \calV} \lambda_V^{(j)}$.
Here, the second inequality follows from the Cauchy-Schwarz inequality.

In addition, we have, for a fixed $\bfV$, for all $\s,\s'$ and $\bfU,\bfU'$,
\begin{align*}
\|\nabla_{\bfU} g(\bfV, \bfU, \s) - \nabla_{\bfU} g(\bfV, \bfU', \s')\|_2 
\leq \sqrt{|\calU|}\left((1 + \alpha_0)\|\bfV\|_F^2 + \bar{\lambda}_U\right)\|\bfU-\bfU'\|_2.
\end{align*}
The derivation is analogous to the case of $\bfV$.

Also, for a fixed $\bfV$, for all $\s,\s'$ and $\bfU,\bfU'$, we have:
\begin{align*}
    \|\nabla_{\s}g(\bfV, \bfU, \s) - \nabla_{\s}g(\bfV, \bfU', \s')\|_2 &= \left\|\lambda_f \bfV^\top\bfV \s-\lambda_f \bfV^\top\bfV\s'\right\|_2 \\
    &= \lambda_f\|\bfV^{\top}\bfV(\s-\s')\|_2 \\
    &\leq \lambda_f\|\bfV^{\top}\bfV\|_F\|\s-\s'\|_2 \\
    &\leq \lambda_f\|\bfV\|_F^2\|\s-\s'\|_2,
\end{align*}
where the first/second inequality follows from Cauchy-Schwarz inequality.
Finally, for fixed $\bfU$ and $\s$, for all $\bfV,\bfV'$, we have:
\begin{align*}
    &\|\nabla_{\s}g(\bfV, \bfU, \s) \!-\! \nabla_{\s}g(\bfV', \bfU, \s)\|_2 \\
    &= \lambda_f\|\bfV^{\top}\bfV\s\!-\!\bfV'^{\top}\bfV'\s\|_2 \\
    &= \lambda_f\|(\bfV^{\top}(\bfV\!-\!\bfV')+(\bfV\!-\!\bfV')^{\top}\bfV')\s\|_2 \\
    &\leq \lambda_f\|\bfV^{\top}(\bfV\!-\!\bfV')\s\|_2+\lambda_f\|(\bfV\!-\!\bfV')^{\top}\bfV'\s\|_2 \\
    &\leq \lambda_f\|\bfV^{\top}(\bfV\!-\!\bfV')\|_F\|\s\|_2+\lambda_f\|(\bfV\!-\!\bfV')^{\top}\bfV'\|_F\|\s\|_2 \\
    &\leq \lambda_f(\|\bfV\|_F + \|\bfV'\|_F)\|\s\|_2\|\bfV\!-\!\bfV'\|_F \\
    &= \lambda_f(\|\bfV\|_F + \|\bfV'\|_F)\|\s\|_2\|\bfV\!-\!\bfV'\|_F,
\end{align*}
where the second/third inequality follows from Cauchy-Schwarz inequality.
\end{proof}